\newif\ifnotes
\newif\ifqip
\newif\ifsubmission

\ifsubmission
\notesfalse
\fi

\ifsubmission
    \documentclass{llncs}
\else
    \documentclass[11pt]{article} 
\fi

\usepackage{amsfonts,amssymb,amsmath}
\ifsubmission \else \usepackage{amsthm} \fi
\usepackage{mathrsfs}
\ifsubmission \else  \usepackage{fullpage} \fi
\usepackage[colorlinks=true,linkcolor=blue,citecolor=blue,breaklinks=true]{hyperref}
\usepackage{breakcites}
\usepackage{algorithm}
\usepackage{algorithmicx}
\usepackage[noend]{algpseudocode}
\usepackage{bm} 
\usepackage{dsfont}
\usepackage{graphicx, float, tikz}
\usepackage{colonequals}  
\usepackage{color}
\usepackage{multicol}
\usepackage{tikz}
\usepackage{fancybox}
\usepackage[breakable, theorems, skins]{tcolorbox}
\usetikzlibrary{matrix,arrows}
\usetikzlibrary{positioning}
\usetikzlibrary{decorations}
\usetikzlibrary{decorations.pathreplacing}
\usetikzlibrary{shapes.geometric, chains, calc, fit, matrix}
\tikzstyle{system}=[rectangle,draw,fill=lightgray,minimum height=0.8cm,minimum
            width=0.8cm,thick]
\tikzstyle{BC}=[system]
\tikzstyle{resource}=[system]
\tikzstyle{RO}=[resource, minimum width=1cm]
\tikzstyle{protocol}=[circle, inner sep=0.7mm, draw]
\tikzstyle{simulator}=[circle, inner sep=0.7mm, draw]
\tikzstyle{memory}=[resource]
\tikzstyle{distinguisher}=[resource,fill=white,minimum width=3.5cm,
minimum height=1.2cm]
\tikzstyle{link}=[]

\usepackage{cleveref}

\usepackage{paralist}
\usepackage{enumerate}
\usepackage{enumitem}
\usepackage[n,
advantage,
operators,
sets,
adversary,
landau ,
probability,
notions,
logic,
ff,
mm,
primitives,
events,
complexity,
asymptotics,
keys]{cryptocode}

\renewcommand{\secpar}{\lambda}
\usepackage{tikz}
\usetikzlibrary{arrows,automata}

\usepackage{braket}

\usepackage[framemethod=tikz]{mdframed}
\usepackage{subcaption}

\usepackage{algpseudocode}

\usepackage{cleveref}

\usepackage[colorinlistoftodos]{todonotes}
\input{general_macros}

\newcommand{\ct}{\mathsf{ct}}

\newcommand{\verkey}{\mathsf{vk}}
\renewcommand{\vk}{\verkey}

\newcommand{\VerKeyGen}{\mathsf{VKGen}}

\newcommand{\Hyb}{\mathsf{Hyb}}

\newcommand{\Sim}{\mathsf{Sim}}

\newcommand{\FHE}{\mathsf{FHE}}
\newcommand{\PKE}{\mathsf{PKE}}

\newcommand{\Can}{\mathsf{Can}}

\newcommand{\Sender}{\mathsf{Sen}}
\newcommand{\Recvr}{\mathsf{Recvr}}

\newcommand{\Accept}{\mathsf{Acc}}
\newcommand{\Reject}{\mathsf{Rej}}

\newcommand{\PoNI}{\mathsf{PoNI}}

\newcommand{\PoNISec}{\PoNI\text{-}\mathsf{Sec}}
\newcommand{\PoNIEncSearch}{\PoNI\text{-}\mathsf{Enc}\text{-}\mathsf{S}}

\newcommand{\PoNIMoE}{\PoNI\text{-}\mathsf{MoE}}
\newcommand{\CDGame}{\mathsf{CD\text{-}Game}}

\newcommand{\coset}{\mathsf{coset}}

\newcommand{\Subspc}{\mathsf{Subspc}}
\newcommand{\Supspc}{\mathsf{Supspc}}

\newcommand{\co}{\mathsf{co}}

\newcommand{\OSP}{\mathsf{OSP}}

\newcommand{\aux}{\mathsf{aux}}

\newcommand{\dimS}{{2\secpar/6}}
\newcommand{\dimT}{{3\secpar/6}}
\newcommand{\dimR}{{\secpar/6}}
\newcommand{\dimSplusT}{{4\secpar/6}}
\newcommand{\dimW}{\dimSplusT}
\newcommand{\dimTminR}{2\secpar/6}
\newcommand{\dimTminusR}{\dimTminR}

\newcommand{\ctr}{\mathsf{ctr}}

\newcommand{\cAuxClass}{k_C}

\newcommand{\msk}{\mathsf{msk}}

\newcommand{\ATI}{\mathsf{ATI}}

\title{Proofs of No Intrusion}
\ifsubmission 
    \author{} 
\else
    \author{Vipul Goyal\thanks{NTT Research.  \texttt{vipul@vipulgoyal.org}} \and Justin Raizes\thanks{NTT Research.   \texttt{justin.raizes@ntt-research.com}}}
\fi
\date{}

\begin{document}

\maketitle

\ifqip
\else
\begin{abstract}
    A central challenge in data security is not just preventing theft, but detecting whether it has occurred. Classically, this is impossible because a perfect copy leaves no evidence. Quantum mechanics, on the other hand, forbids general duplication, opening up new possibilities.

We introduce \emph{Proofs of No Intrusion}, which enable a classical client to remotely test whether a quantum server has been hacked and the client's data stolen. 
Crucially, the test \emph{does not destroy the data being tested}, avoiding the need to store a backup elsewhere.
We define and construct proofs of no intrusion for ciphertexts assuming fully homomorphic encryption.
Additionally, we show how to equip several constructions of unclonable primitives with proofs of non-intrusion, such as unclonable decryption keys and signature tokens. Conceptually, proofs of non-intrusion can be defined for essentially any unclonable primitive.

At the heart of our techniques is a new method for non-destructively testing coset states with classical communication. It can be viewed as a non-destructive proof of knowledge of a measurement result of the coset state.
\end{abstract}
\fi

\ifqip
\else
\ifsubmission
\else
\clearpage
\tableofcontents
\pagebreak
\ifnotes
\input{Notes/02-01-overview-overview}
\pagebreak
\fi
\fi


\section{Introduction}

Suppose a client wanted to outsource the storage of their data to an external provider. Later, they find out that the external provider suffered a data breach. Is there a way for the client to check whether their own data was stolen in the hack?

This question is not only a theoretical curiosity. All 50 states in the United States have laws that require the affected parties to promptly be notified in the event of a data breach that affects their personal information. 
Unfortunately, enforcing said laws is tricky because perfectly detecting stolen data is an impossible task in the classical setting. The hacker could have simply made an exact copy of everything on the storage provider's servers. However, in the quantum setting, the no-cloning theorem opens up new possibilities. 

\c{C}akan, Goyal, Liu-Zhang, and Ribeiro \cite{TCC:CGLR24} showed how a server can utilize quantum mechanics to locally test whether its data has been stolen. 
Their approach builds on the notion of certified deletion \cite{TCC:BroIsl20,C:BarKhu23}, which allows a server to irreversibly delete a quantum ciphertext and produce a certificate proving that the ciphertext can no longer be decrypted -- even if the decryption key later leaks or the encryption scheme is broken. 
\cite{TCC:CGLR24} observed that if a server can coherently generate such a certificate, then no copy of the ciphertext could exist.
Crucially, if the certificate is checked coherently, then it can be uncomputed, leaving the data intact.
If one is willing to assume quantum communications to a quantum client, then the client can non-destructively test the data themselves.

Of course, even a quantum world, an everyday client who does not own a quantum computer should be protected. Unfortunately, sending a deletion certificate to a classical client necessarily deletes the data.
For a client who outsourced their data storage to avoid having to store a local copy, destroying the server's copy means the data is lost forever. 

Similar tensions arise in software leasing and anti-piracy. Secure software leasing~\cite{EC:AnaLaP21} allows an owner to lease out software, then check that the software has not been copied at the \emph{end} of the lease. However, they might want to check that the software has not been copied in the \emph{middle} of the leasing period, without ending the lease early by destroying the program. If the test can be done using a classical verifier, the owner can reserve their expensive quantum computer for new leases and manage existing customers using more ubiquitous classical computers.

Motivated by these examples, we ask
\begin{quote}
    \centering
    \emph{Is it possible for a quantum server to prove to a classical client that the client's data has not been stolen, without destroying the data?}
\end{quote}
As an answer, we introduce Proofs of No-Intrusion (PoNIs).
A PoNI is an interactive protocol in which a quantum server convinces a classical client that no adversary has stolen usable information about the stored data. Crucially, the data remains intact after the proof. 
Similarly to classical verification of quantum computation (CVQC), PoNIs can be seen as developing a classical client's ability to fruitfully use an untrusted (or semi-trusted) quantum cloud. CVQC allows a classical verifier to test if a quantum computation was carried out correctly, while PoNIs allow a classical verifier to test if quantum data has been stored securely.

\subsection{Our Results.}

We initiate the formal study of Proofs of No-Intrusion and construct several primitives that support them. 

Our first construction addresses the outsourced storage scenario. We design an encryption scheme in which a server can repeatedly prove to a client that no information useful for decryption has leaked. Informally, if the server can produce an accepting PoNI, then no adversary has already copied anything sufficient to recover the plaintext -- even if the decryption key later leaks or the encryption scheme is broken. Even in these extreme circumstances, it is as if the server was never hacked.
Moreover, the ciphertext remains decryptable after the proof, and the client may test the same ciphertext as many times as it wants.

\begin{theorem}[Informal]
    Assuming fully homomorphic encryption (FHE), there exists encryption with Proofs of No-Intrusion for ciphertexts.
\end{theorem}

As our main technical ingredient, we show how to construct PoNIs for coset states. 
Coset states provide very strong unclonability properties and have been a core tool for many constructions of unclonable primitives~\cite{STOC:AarChr12,C:CLLZ21,C:AKLLZ22,EC:BGKMRR24,EPRINT:CakGoyRai24}.
Given the widespread use of coset states for unclonability, we believe our new tool to be of independent interest. 
We elaborate on the precise properties of our PoNI for coset states in \Cref{sec:tech-over}.

If one is willing to assume indistinguishability obfuscation, certain techniques developed in prior work for unclonability can be reinterpreted directly as PoNIs for coset states. The core technical contribution of this work is to construct PoNIs for coset states from significantly weaker assumptions.

\begin{theorem}[Informal]
    Assuming fully homomorphic encryption, there exists a Proof of No-Intrusion for coset states.
\end{theorem}

PoNIs for coset states naturally extend to other unclonable primitives that rely on coset states. 
By extending existing constructions which use coset states for unclonability, we can equip several other primitives with PoNIs.

For example, unclonable decryption keys~\cite{EPRINT:GeoZha20} naturally have a PoNI: ask the prover to decrypt a random message. If the prover is able to do so, then no one else can possess a decryption key.
As another example, consider signature tokens~\cite{Q:BS23}, which allow the holder to sign one, and only one, message of their choice on behalf of a delegating party. The delegator could test that the intended signer still holds the token by asking them to sign a message, but this consumes the token. Instead, by using a PoNI for coset states in conjunction with \cite{C:CLLZ21}'s construction where tokens are coset states, the delegator can test the token \emph{without} consuming it prematurely.

\begin{theorem}[Informal]
    Assuming post-quantum indistinguishability obfuscation and one-way functions, there exists encryption with Proofs of No-Intrusion for decryption keys.
    
    Assuming sub-exponentially secure post-quantum indistinguishability obfuscation and one-way functions, there exist signature tokens with Proofs of No-Intrusion.
\end{theorem}

The underlying construction of these primitives requires heavyweight tools like indistinguishability obfuscation, regardless of how the PoNI is constructed.
Nonetheless, we view our construction of PoNIs for coset states from FHE as evidence that PoNIs can be significantly lighter-weight than unclonability solutions for the same primitives.

\subsection{Related Notions}
\label{related}

PoNIs are closely related to both certified deletion and unclonable cryptography. All three exploit the unclonability of quantum information, but they differ in the precise security guarantee and have distinct use cases. For concreteness, we compare the three in the context of encryption.
\begin{itemize}
    \item \textbf{Certified deletion} allows a party to provably delete a ciphertext. If the proof is valid, then the encrypted message cannot be recovered even using the secret key or unbounded computational power. Often, the proof is a single classical message, called a certificate.

    \item \textbf{Unclonable encryption} (or copy protection) ensures that a ciphertext cannot be copied. If a ciphertext is split in two, then at most \emph{one} of the copies can be decrypted using the secret key.

    \item \textbf{Proofs of No Intrusion} allow a server to prove to an external classical client that some data has not been stolen. If the client accepts, they know that no hacker who previously compromised the server can recover the encrypted message even using the secret key or unbounded computational power. Moreover, the ciphertext remains intact after the protocol.
\end{itemize}
The security guarantee for PoNIs is most similar to certified deletion. Both require that an adversary cannot simultaneously produce a proof and a state that can be decrypted. 
In contrast, unclonable encryption replaces the proof by a more flexible task: recover the message using the secret key. Relaxing the adversary's task makes security harder to achieve, so unclonable cryptography often relies on significantly stronger assumptions than certified deletion.

Aside from the security guarantee, PoNIs are also intended for a different use case than either certified deletion or unclonable encryption. Certified deletion is only non-destructive if done locally (using public verifiability), while unclonable encryption also only tests the copy locally. If quantum channels were easily available, then the states could be sent to an external quantum verifier, who tests them before sending them back. However, this requires both quantum communication and for the verifier to be quantum themselves.

Preserving the data also opens up the possibility of \emph{repeatedly} testing the same data. 
In this scenario, it becomes important to ensure that prior tests do not compromise the security guarantees of later tests. For example, if the proof were non-interactive, a cheating prover could replay the same proof repeatedly, regardless of whether the data has been stolen.
In contrast, repeatedly testing the same data does not even make sense for certified deletion or unclonable encryption.

\section{Technical Overview}\label{sec:tech-over}

Throughout this overview, we will focus on PoNIs for ciphertexts and the techniques used to build it. For the sake of exposition, we leave the details of other applications to the technical sections.

\subsection{Definitions: PoNIs for Encryption}

An encryption scheme allows a sender to transmit a message which is hidden from anyone without the secret key. In many cases, ciphertexts are stored externally. 
For example, a client might pay a storage provider to externally store its data, but encrypt the data for privacy.

The goal of equipping an encryption scheme with a proof of no-intrusion is to allow the sender to test whether the encrypted message has been exfiltrated from the server. Supporting this, we augment the syntax of the encryption scheme to use a verification key $\verkey$ in the encryption step:
\[
    \ct \gets \Enc(\pk, m, \diff{\verkey})
\]
Later, the sender can use the verification key to test the server in an interactive protocol 
\[
    \Accept/\Reject \gets \PoNI\langle P(\ct), V(\verkey)\rangle
\]
where the sender is the verifier $V$ and the server is the prover $P$.

\paragraph{Minimizing Client Storage.} 
For PoNIs to be useful for outsourced storage, the client's local storage should not grow with the size of the outsourced database. Supporting this, we require two additional properties aside from security:
\begin{itemize}
    \item \textbf{Key Reusability.} A single $\verkey$ can be used to encrypt many messages. Otherwise the client would need to store as much verification material as outsourced data.
    \item \textbf{State Preservation.} After testing a ciphertext, the residual state should still decrypt to the same data. If each PoNI destroyed the ciphertext, the client would need to store a copy of the data locally to avoid losing it, again undermining the point of outsourcing storage.
\end{itemize}

\paragraph{Security.}
Search security for PoNIs can be seen as a ``cloning game'' in the general framework of \cite{C:AnaKalLiu23}.
\begin{enumerate}
    \item First, the storage provider $P_1$ receives the public key $\pk$ and a ciphertext $\ct \gets \Enc(\pk, m, \verkey)$ encrypting a random message $m$.
    \item The verifier $V$ periodically asks $P_1$ for a proof of no intrusion $\PoNI\langle P_1(\ct), V(\verkey)\rangle$. 
    
    If any execution rejects, $V$ can conclude that the data has been stolen, so $P$ immediately loses the security game. As a result, an adversarial prover never sees any rejecting executions unless it has already lost the game.
    
    \item Eventually, $P_1$ is ``hacked'', splitting its state into two registers $\regP$ and $\regH$. The hacker $H$ holds $\regH$ and the storage provider $P_2$ retains $\regP$.
    
    \item $V$ asks $P_2$ for a final proof of no intrusion $\PoNI\langle P_2(\regP), V(\verkey)\rangle$.
    
    \item Meanwhile, $H$ attempts to decrypt the stolen ciphertext in $\regH$. They are given the decryption and verification keys $(\sk, \vk)$ and unbounded computational power while doing so.

    \item The adversary wins if $V$ accepts $P$'s final $\PoNI$ and \emph{simultaneously} $H$ correctly guesses $m$.
\end{enumerate}
Security requires that any QPT $P = (P_1, P_2)$ and unbounded $H$ can only win with negligible probability.
Importantly, the same $\verkey$ is used to test the prover multiple times. Earlier $\PoNI$ executions should not compromise the security of later ones.

\paragraph{Comparison to Unclonable Encryption.} The definition of PoNI security has some similarities to that of unclonable encryption. The primary conceptual difference is in unclonable encryption, \emph{both} parties are trying to guess the message $m$. In contrast, in PoNI security, one of the two parties is instead attempting to convince a verifier in a potentially interactive protocol. The distinction is magnified if we consider indistinguishable security for both primitives. Indistinguishable-unclonable encryption requires both adversaries to solve decision games, while PoNI decision security would require the hacker to solve a decision game and the prover to solve a \emph{search} game.

\paragraph{Search versus Decision.} 
Because a PoNI decision security game would mix a decision game with a search game, the definition is somewhat more involved. 
Intuitively, the PoNI should not accept simultaneously with the hacker ``being good at guessing $m$''. 
Similar scenarios have been considered in certified deletion, which we base our decisional PoNI definition on.
We refer the reader to \Cref{sec:poni-enc:defs} for more details on how to define decisional PoNI security.

\subsection{PoNIs for Coset States.}

Many constructions of unclonable primitives derive their unclonability properties from subspace states~\cite{STOC:AarChr12} or their generalization, coset states~\cite{C:CLLZ21}. A coset state $\ket{S_{x,z}}$ is a superposition over all elements of a coset $S+x$ of a subspace $S \subset \bbF_2^\secpar$, with some additional phase information to also encode $S^\perp + z$:
\[
    \ket{S_{x,z}} 
    = Z^zX^x \ket{S} 
    = Z^z\frac{1}{\sqrt{|S|}} \sum_{s\in S} \ket{s+x}
\]
Towards building PoNIs for ciphertexts, we will aim to build PoNIs for coset states. 
A PoNI for coset states ensures that if an adversary splits a random $\ket{S_{x,z}}$ into two registers and distributes them to parties $P$ and $H$, then $P$ cannot pass a $\PoNI$ simultaneously with $H$ finding a vector in $S^\perp + z$.
Other security properties can also be obtained by changing $H$'s task, such as requiring it to find a vector in $S + x$.

Crucially, the PoNI should not destroy the coset state. Otherwise, the verifier could destructively test $\ket{S_{x,z}}$ by asking for a computational basis measurement and checking that the prover returns an element of $S + x$. In that case, direct product hardness~\cite{C:CLLZ21} would rule out $H$ simultaneously knowing a vector in $S^\perp + z$, but $\ket{S_{x,z}}$ would collapse to a computational basis state.
It will be useful to think of the description $(S,x)$ as the verification key, although specific constructions may need more or less information about $\ket{S_{x,z}}$ to test it.

\paragraph{From Coset States to Encryption.}
PoNIs for coset states smoothly enable PoNIs for ciphertexts. A message $m$ can be encrypted as a pair
\[
    \ket{S_{x,z}},\ \Enc(S, m\oplus z)
\]
using any off-the-shelf encryption scheme, similarly to \cite{EC:BGKMRR24}. 
To avoid storing a separate $(S,x)$ to verify each ciphertext, the sender generates $(S,x)$ using a pseudorandom function. Then the verification key is a single PRF seed that can be used for encrypting arbitrarily many messages.

Decryption works by first decrypting $S$ and $m\oplus z$, then measuring $\ket{S_{x,z}}$ in the Hadamard basis. The measurement outcome $v \in S^\perp + z$ allows one to find $z$ (given $S$) and finally unmask $m$.\footnote{For the sake of exposition, we assume it is possible to find $z$ given $S$ and $v\in S^\perp + z$. In the technical sections, we mask $m$ using a canonical element $\Can_S(z)$ of the coset $S^\perp+z$.} 

The sender can test the ciphertext by running the coset state PoNI for $\ket{S_{x,z}}$. Even if the adversary breaks the off-the-shelf encryption scheme to recover $m\oplus z$, they must still find $z$ to recover $m$. The PoNI for $\ket{S_{x,z}}$ prevents this.

Looking ahead, our construction of PoNIs for coset states will add a random $Z$ error $z'$ to $\ket{S_{x,z}}$, known only to the verifier. To ensure that the ciphertext still decrypts to the same message, it will also contain a fully homomorphic encryption of a correction factor (initially $0$). After each PoNI, the verifier will send a fully homomorphic encryption $\FHE(z')$ of the correction factor and the prover will homomorphically update the correction factor. During decryption, the decryptor can use the aggregate correction factor $z_{\mathsf{cor}}$ to unmask $m$ from $m\oplus z \oplus z_{\mathsf{cor}}$.

\paragraph{Proofs of No-Intrusion from Indistinguishability Obfuscation.} 
Prior works can be interpreted as having implicitly constructed PoNIs for coset states using indistinguishability obfuscation (iO). Adapting these ideas, there is a simple two-message protocol.
\begin{enumerate}
    \item The verifier picks a random message $m$ and obfuscates the circuit $C_{m, S+x}$ which takes as input $v$, then outputs $m$ if $v \in S+x$.
    
    \item The prover coherently evaluates the obfuscated circuit on $\ket{S_{x,z}}$ in the computational basis, then measures the output $m$ and sends it to the verifier. 
    
    \item The verifier checks that it received $m$.
\end{enumerate}

Intuitively, the only way to obtain $m$ is to evaluate the obfuscated circuit on a vector in $S+x$. \cite{C:CLLZ21,eprint:KY25} formalize this intuition by showing that it is possible to extract an element of $ S+x$ from any prover who finds $m$. Then they leverage a coset monogamy of entanglement result from~\cite{C:CLLZ21,Q:CV22} to show that a second adversary cannot simultaneously find a vector in $S^\perp + z$.

This construction has a few nice properties. First, it doesn't disturb $\ket{S_{x,z}}$ at all because the measurement result $m$ is deterministic. 
Second, publishing an obfuscation of $C_{1, S+x}$ allows \emph{anyone} to implement $C_{m,S+x}$ and so act as the verifier.

\paragraph{Towards Weaker Assumptions.} 
Unfortunately, post-quantum indistinguishability obfuscation is quite a strong assumption. 
The main technical contribution of this work to construct PoNIs for coset states under a weaker assumption: oblivious state preparation (OSP)~\cite{AC:CCKW19,C:BK25} for coset states.

Informally, an oblivious state preparation protocol for coset states allows a classical sender to instruct a quantum receiver to produce a coset state $\ket{S_{x,z}}$ for a subspace $S$ of the sender's choice, without revealing any information about $S$.
Shmueli~\cite{STOC:Shmueli22} showed how to construct an OSP for coset states using hybrid quantum fully homomorphic encryption (FHE), which can be based on any post-quantum FHE~\cite{C:GupVai24}. First, the sender sends a hybrid encryption of a description of $S$:
\[
    \FHE(x,z),\ X^x Z^z\ket{\mathsf{descr}(S)}.
\]
Then, the receiver homomorphically computes the subspace state $\ket{S}$, producing a state underneath a fresh one-time pad:
\[
    \FHE(x',z'),\ X^{x'}Z^{z'} \ket{S} = \ket{S_{x',z'}}
\]
Finally, the receiver sends $\FHE(x',z')$ back to the sender so that the sender can learn $x'$ and $z'$.
It may be useful to keep this preparation protocol in mind for intuition about the properties of OSP protocols later.

\paragraph{PoNIs from Oblivious State Preparation.}
In our construction of proofs of no-intrusion for coset states, the verifier uses a description $(T,x_T)$ of a random supercoset $T+x_T\supset S + x$ to remotely verify the state $\ket{S_{x,z}}$. If $S+x$ is generated via PRF, then the verification key is the PRF seed.
\begin{enumerate}
    \item The verifier uses OSP to send $\ket{T_{x',z'}}$ to the prover, learning $(x', z')$ in the process.
    
    \item The prover applies a CNOT from $\ket{S_{x,z}}$ to $\ket{T_{x',z'}}$. Since $S\subset T$, this causes a phase kickback and the result is
    \[
        \ket{S_{x,z-z'}} \otimes \ket{T_{x+x',z'}}
    \]
    The prover measures $\ket{T_{x+x',z'}}$ in the computational basis to obtain $v\in T+x + x'$, then sends $v$ to the verifier.
    
    \item The verifier checks that $v - x' \in T + x_T$.
\end{enumerate}

The reader may notice that at the end of the protocol, $\ket{S_{x,z}}$ has been modified by a random $Z$ operation, known only to the verifier. 
This error is relatively benign, since the original $z$ is still information-theoretically determined. 
As we saw previously, a coset state PoNI with $Z$ errors can still be used for constructing PoNIs for ciphertexts, without changing the message that the ciphertext decrypts to.
Furthermore, $Z$ errors do not prevent additional PoNI executions, since those only test the computational basis.
We provide more discussion in \Cref{sec:overview-future} about possibility of removing the $Z$ error and the difficulties therein.

\paragraph{A Comment on Fully Homomorphic Encryption.} 
It may be tempting to attempt to build PoNIs directly from FHE, without going through OSP.
For instance, the prover could use $\FHE(S)$ to homomorphically evaluate a canonical element of $S + x'$ for every $\ket{x'}$ in the support of $\ket{S_{x,z}}$. Because every $x' \in S+ x$ maps to the same canonical element, the evaluated message is independent of the specific $x'$.
One might therefore hope that the evaluated ciphertext is deterministic and thus safe to measure.

Unfortunately, the security of FHE \emph{forces} the evaluated ciphertext entangle with $\ket{S_{x,z}}$ anyway. 
For example, if the ciphertext instead encrypted the trivial subspace $\{0\}$, the canonical element of $\{0\}+x'$ is $x'$. In this case, the evaluated ciphertext must entangle with each $\ket{x'}$. Moreover, this must also apply to the case of encrypting $S$ since $\FHE(S)$ and $\FHE(\{0\})$ are indistinguishable.

\subsection{Security}

We give an overview of PoNI security through the lens of encryption, for the sake of exposition. We prove a more general statement in the technical sections. 
To prove security, we reduce to existing hardness results for simultaneously finding vectors in $S+x$ and $S^\perp + z$. Certified deletion for coset states~\cite{EC:BGKMRR24} considers the following scenario:
\begin{enumerate}
    \item  Sample psuedo-random $S$, $x$, and random $z$. QPT adversary $A$ gets $\ket{S_{x,z}}$ and some auxiliary information $\aux_A(S,x)$, then splits their state into two registers $\regR_B$ and $\regR_C$.
    \item QPT adversary $B$ gets $\regR_B$, then attempts to output a vector in $S + x$.
    \item Unbounded adversary $C$ gets $\regR_C$ and attempts to output a vector in $S^\perp + z$.
    \item The adversary wins if $B$ outputs a vector in $S+x$ and $C$ outputs a vector in $S^\perp + z$.
\end{enumerate}
It guarantees that if $\aux_A(S,x)$ is ``subspace-hiding'', then no QPT $(A,B)$ and unbounded $C$ can win the game with better than negligible probability. Informally, subspace-hiding means that it computationally hides everything except for random super-cosets $T + x_T \supset S + x$ and  $R^\perp +z_{R^\perp} \supset S^\perp + z$. Although \cite{EC:BGKMRR24}'s original theorem only considers truly random $(S, x)$, a careful examination of their proof shows that security holds even when $(S,x)$ are \emph{pseudo}-random.

Imagine that $A$ was the pre-split prover $P_1$ in the ciphertext PoNI security game, where $\aux_A$ is its pre-split view, that $B$ was the post-split prover $P_2$, and that $C$ was the hacker $H$. 
If $\aux_A$ were subspace-hiding and we could show how to extract an element of $S+x$ from a $P$ which gives an accepting PoNI, then certified deletion rules out $H$ simultaneously finding an element of $S^\perp +z$.

\paragraph{The Adversary's View.}
As $P_1$'s pre-split view, $\aux_A$ includes both a ciphertext $(\ket{S_{x,z}}, \Enc(S, z\oplus m))$ and the leftover prover view from many PoNI executions. 
Ignoring the decision bits for these executions can only increase the adversary's winning probability because the adversary loses if any of them reject. In this case, the prover's view only consists of the receiver view in an OSP. For concreteness, the reader can imagine using Shmueli's OSP, in which case $\aux_A$ is 
\[
    \Enc(S, z\oplus m),\ \mathsf{QFHE}_{\pk_1}(T), \dots, \mathsf{QFHE}_{\pk_n}(T).
\]
The semantic security of encryption immediately implies that $\aux_A$ is subspace hiding. In fact, it reveals even less than what is allowed. Certified deletion would still hold even if $\aux_A$ also included $T$, $x_T$, and $R$ for random $T\supset S$, $x_T \in T + x$, and $R^\perp \supset S^\perp$. Looking forward, these will play a crucial role in the reduction.

\paragraph{High Level Extraction Strategy.}
The main challenge is how to extract an element of $S+x$ from $P_2$ \emph{simultaneously} with $H$ finding an element of $S^\perp + z$ from any $(P_1, P_2, H)$ that violate PoNI security. There are two components to this: 1) how to extract from $P_2$ and 2) showing that extraction succeeds simultaneously with $H$'s success.

First, let us address the high-level strategy for extracting from $P_2$. 
The only black-box way to interact with $P_2$ is to ask it for a PoNI using $T\supset S$, during which $P_2$ returns a vector $v_P$ and the verifier learns $x_{\OSP}$ when $P_2$ prepares $\ket{T_{x_{\OSP}, z_{\OSP}}}$ during the $\OSP$ protocol.
If the PoNI is accepting, then $v_P - x_{\OSP} \in T + x_T$. 
To find an element of $S+x$, we observe that cosets over finite fields obey a property similar to the Chinese remainder theorem: given
\begin{align*}
    x_1 &\in T_1 + x
    \\
    x_2 &\in T_2 + x
\end{align*}
along with $T_1$ and $T_2$, it is efficient to compute an element of $(T_1 \intersect T_2) + x$. 
With this in mind, the extractor can ask $P_2$ to answer a PoNI using $T'$ such that $(T' \intersect T) = S$. If it learns a vector $v_P - x_{\OSP} \in T' + x$, then it can combine that with $T$ and $x_T\in T + x$, which are provided for free in $\aux_A$, to find an element of $S + x$.

Unfortunately, the extractor is not allowed to know such a $T'$. Knowledge of both $T$ and $T'$ would reveal $S$, violating the premise of certified deletion.\footnote{Monogamy of entanglement~\cite{C:CLLZ21,Q:CV22} would make this step easier by revealing $S$ after the split, but unfortunately does not allow $C$ to learn $\sk$, unless we are willing to assume iO-equivalent assumptions such as functional encryption~\cite{TCC:AnaKal21,FOCS:BitVai15,C:AnaJai15}. Additionally, it never allows unbounded $C$ when $A$ sees $\Enc(S)$.}
A natural next approach is to try $(T' \intersect T) \subset S$. 
The subtle issue with this approach is that $T'$ would \emph{never} occur in a real execution, because $T\not\supset S$. 
If $P_2$ were to answer the $T'$ PoNI honestly, then it would add an element of $S+x$ to an element of $T'+x_{\OSP}$ and return the result. In this case, the extractor learns $(v_P - x_{\OSP}) \in (T' + S)+x$, which a little algebra shows is enough for the coset Chinese remainder theorem even just using $T'$ instead of the full $T' + S$.
Of course, $P_2$ is not honest. 
It is not clear that $P_2$ will answer such a $T'$ the same way as the real $T$.
One might hope that $\OSP$ hides $T'$, and so $P_2$ cannot change its behavior when given $T'$. However, $x_{\OSP}$ is also hidden from $P_2$, so there are no guarantees about $v_P - x_{\OSP}$.
As an example, in Shmeuli's OSP, $x_{\OSP}$ is obtained by decrypting a QFHE ciphertext which depends on a QFHE encryption of $T'$.

As a solution, we carefully design a $T'$ which is computationally indistinguishable from a real one even given $\ket{S_{x,z}}$, $\aux_A$, and some new information $W = (T' + S)$ and $x_W \in W + x$. Together, $W$ and $x_W$ allow us to check that $v_P - x_{\OSP} \in W + x$, so using an indistinguishable $T'$ in the PoNI cannot noticeably change the probability of this predicate occurring.
Interestingly, although prior works treated $R^\perp$ only as a way to hide $S^\perp$, $R$ turns out to play a crucial role in finding an appropriate $T'$. We describe $T'$ in more detail shortly hereafter.

\paragraph{Simultaneous Extraction.}
The second issue to address is showing that extraction succeeds simultaneously with $H$'s success. Although $T'$ will be indistinguishable to $P_2$, it is \emph{not} indistinguishable to the unbounded $H$ who also receives $\vk$ and $\sk$. 
A similar question about asymmetrically changing challenge distributions in cloning games is at the heart of the open problem of indistinguishable-unclonable encryption~\cite{TQC:BL20,TCC:AnaKal21,C:AKLLZ22,C:AnaKalLiu23,C:AnaBeh24}.
We prove a general theorem for such simultaneous search tasks where $B$'s distribution changes in a way that is indistinguishable to $B$, but distinguishable to $C$. 
At a high level, we convert the decision task of distinguishing, for example, $T$ from $T'$, into the search task of finding the correct $T$ among a large polynomial number of candidate $T'$. Then, we decompose the joint state of $B$ and $C$ into a sum of orthogonal states corresponding to $C$'s success or failure, which cannot interfere with one another on $B$'s side. Finally, we argue that in the state corresponding to $C$'s success, $B$ is able to find the correct $T$, contradicting search security. 

If we had reduced to a decisional task, then the two branches corresponding to $C$'s success could respectively succeed with probabilities $>1/2$ and $<1/2$, canceling each others' advantage even without interfering quantumly. Crucially, in a search game, success is absolute and cannot cancel, so successful search in one branch is sufficient.

\paragraph{The Alternate Challenge.}
The alternate challenge turns out to be very simple: use a random superspace $T_R \supset R$ with the same dimension as $T$. 
Despite $T_R$'s simplicity, seeing why it is sufficient and why $R$ is necessary requires some additional insights.
There is a tension between needing to find $v_P - x_{\OSP} \in (T_R + S) + x$ and the need for the alternate challenge to be indistinguishable while checking this condition. This tension ultimately forces the alternate challenge to intersect significantly with $S$, resulting in the use of $R$.

\paragraph{Disconnecting $T$ from $S$.} 
As an initial step, we follow the well-established strategy of replacing a coset membership check using a larger subspace: check that $v_P - x_{\OSP} \in W + x$ for a random $W \supset T$. If the PoNI using $T$ would pass, then this condition is satisfied. 
Furthermore, the condition can be checked using knowledge of $W$ and an $x_W \in W + x$.
The use of $W$ introduces some slack for us to exploit. As an additional simplifying step, we may invoke the psuedorandomness of $(S, x)$ and replace $\aux_A$ with encryptions of $0$, so that we only need focus on random $S$, $x$, and $T$.

The first conceptual leap is to partially disconnect the superspace used in the PoNI from $S$ by replacing $T \supset S$ by an independently random $T_W \subset W$. The slack from moving to a larger subspace ensures that $T_W$ likely will not contain $S$, moving us closer to a challenge that is both useful for the coset remainder theorem and that certified deletion allows the extractor to see. 
As we will see shortly, $T_W$ is not yet enough for this purpose, but for now, let us focus on the question of why $P_2$ will answer a PoNI using $T_W$ with $v_P - x_\OSP \in W + x$.

For this to be the case, $T_W$ needs to be indistinguishable from $T$ \emph{even conditioned} on seeing $W$ and $x_W$, which are required to check that $v_P - x_\OSP \in W + x$, and $\ket{S_{x,z}}$, which is required to run the PoNI in the first place (we may ignore $\aux_A$ since it just encrypts $0$s).
For any fixed $W$, the offsets $x \in W + x_W$ and $z \in \bbF_2^\secpar$ act as one-time pads for $\ket{S}$ (subject to $S\subset W$) so the mixed state $\rho_{S}$ over $\ket{S_{x,z}}$ turns out to depend only on $W$ and $x_W$. Without $S$ to anchor it, $T$ is just a random subspace of $W$, the same as $T_W$.

Unfortunately, $T_W$ still isn't useful for finding an element of $S+x$.
It turns out that the partial version of the coset Chinese remainder theorem really does require an element of $(T' + S) + x$ when using $T'$ as the input instead of $T'+S$.
If $W$ were very large, then $T_W$ intersects $T$ trivially almost certainly,
satisfying $(T' \intersect T) \subset S$. However, then $W \neq T_W + S$, so the algorithm would fail anyway when using $v_P - x_{\OSP} \in W + x$. 
On the other hand, if $W$ were small enough that $W = T_W + S$, then $T_W \intersect T$ would likely contain elements outside of $S$, and so the algorithm would return an element outside of $S+x$. 
Overall, the requirement that $W = T_W + S$ and fact that $W\supset T$ forces $T$ and $T_W$ to be strongly connected.

\paragraph{Disconnecting $T_W$ from $T$.}
Our second conceptual leap is to disconnect $T_W$ from $T$ by sampling $W$ implicitly using $T_R$. Explicitly, imagine replacing $(W,T_W)$ by $(T_R + S, T_R)$. 
Now $(T_R \intersect T) = R\subset S$ with overwhelming probability and $W=T_R+S$, enabling the partial coset Chinese remainder theorem. 
We show that these distributions are statistically close for any $S$.

Showing that $T_R +S$ is distributed similarly to $W$ is not too hard; $T_R + S$ is also a random superspace of $S$ with dimension $\dim(S) + \dim(T_R) - \dim(R)$, which we set $= \dim(W)$, with overwhelming probability. 
Showing that $T_R \approx_s T_W$ conditioned on $W = T_R + S$ is slightly more involved. 
Conditioned on observing $W = T_R + S$, $T_R$ is uniform over the superspaces of $R$ in this event space. 
$T_W$ is uniformly distributed over subspaces of $W$, so we only need show that it satisfies both of these events with overwhelming probability.

To ensure that $T_W$ satisfies these two conditions, we carefully tune the dimensions of all four subspaces $W$, $T_W$, $S$, and $R$. 
$T_W$ is independent of $S$, modulo containment of both in $W$. To guarantee that $T_W$ and $S$ span $W$ with overwhelming probability, we set $\dim(T_W) + \dim(S) \gg \dim(W)$. 
As a byproduct, these dimensional parameters force $T_W$ and $S$ to have a significant intersection.
Because of this, the extractor's alternate challenge \emph{must} contain a significant subspace of $S$. Fortunately, certified deletion does allow the extractor to know a random $R \subset S$, which can be used to seed $T_R$.

\subsection{Open Questions}\label{sec:overview-future}

This work initiates the study of proofs of no-intrusion and answers some initial feasibility questions, but there are many more open questions in this new direction.

\begin{itemize}
    \item \textbf{Other Primitives.} A wide variety of primitives have been studied in the context of unclonability, such as decryption keys, money, or even general software,~\cite{CIC:A09,STOC:AarChr12,C:ALLZZ21,C:AnaBeh24,TCC:GoyMalRai24,CiC:BroKarLor24}. Generally, any of these primitives can also be considered for proofs of no-intrusion by substituting a proof of no-intrusion for one of the copies an adversary is attempting to make. Which primitives can we construct proofs of no-intrusion for?

    \item \textbf{Other Primitives with Weaker Assumptions.} 
    As a related question, can PoNIs for other primitives be based on weaker assumptions than those required to make them unclonable?
    
    Our results suggest that this is possible. Our proofs of no-intrusion for coset states can be based on LWE, whereas tests for possession of coset states in prior works relied on very strong assumptions such as indistinguishability obfuscation or extractable witness encryption.

    \item \textbf{Decisional Security and a Bridge to Indistinguishable-Unclonable Encryption.} 
    One of the limitations of our analysis is a polynomial loss that occurs when reasoning about simultaneous search, limiting our applications to search security. Is it possible to achieve security where the prover attempts to solve a proof of no-intrusion while the hacker attempts to decide a bit?

    The question of decisional security seems closely related to the open question of indistinguishable-unclonable encryption, but seems easier because only one side faces a decisional task. We hope that any progress towards decisional PoNI security will provide new insights into this exciting open problem.

    \item \textbf{Removing Pauli Errors from LWE-Based Constructions.}
    Our proof of no-intrusion for coset states causes a $Z$ error on the tested state. 
    Is there a construction which leaves the state essentially unmodified, without resorting to indistinguishability obfuscation?

    A \emph{temporary} Pauli error seems inherently necessary to hide the extractor measuring the prover's state. However, it may be possible to reveal and correct the error after the extraction window. Unfortunately, revealing the $Z$ error in our current protocol compromises the security of the $\OSP$, revealing $T$ and thus ruining the security of future PoNI executions.
\end{itemize}
\ifsubmission \else \section{Preliminaries}

\subsection{Algebra}\label{sec:prelims-algebra}

Let $S\subset \bbF_2^n$ be a subspace. $\bbF_2^n$ can be partitioned into cosets $\{S+x:x\in \bbF_2^n\}$ of $S$. It is useful to identify the cosets by a single ``canonical'' element $x\in \bbF_2^n$. We let $\co(S)$ be the set of canonical elements of cosets of $S$. In particular, we choose $\co(S)$ to be a subspace of $\bbF_2^n$ so that $\co(S)$ is isomorphic to the quotient space $\bbF_2^n/S$.

There may be many possible choices of $\co(S)$ for each $S$. For example, one way to choose $\co(S)$ is to fix an arbitrary basis ${b_i}_{i\in[n]}$ of $\bbF_2^n$ where the first $k$ elements $\{b_i\}_{i\in [k]}$ form a basis for $S$. Then one can choose $\co(S) \coloneqq \Span(\{b_i\}_{i\in [k+1,n]})$ to be the span of the remaining basis vectors.

In this example, it is easy to compute the canonical representative of a coset $S+x$ given any element $x' \in S+x$. Simply compute the basis decomposition
\[
    x' = \sum_{i=1}^{n} c_i b_i
\]
and output
\[
    \Can_S(x') \coloneqq \sum_{i = k+1}^n c_i b_i
\]

We also define notation for the set of superspaces and subspaces of a vector space $S \subset \bbF_2^n$, with specified dimension.
\begin{align*}
    \Subspc(S, k) &\coloneqq \left\{T: \dim(T) = k \text{ and } S \subset T \subset \bbF_2^n\right\}
    \\
    \Supspc(S, k) &\coloneqq \left\{R: \dim(R) = k \text{ and } R\subset S\right\}
\end{align*}
A uniform sample from these sets is denoted as $T \gets \Subspc(S,k)$.

When an algorithm takes as input a coset $S+x$, we mean that it takes as input a description of the coset. This consists of a description of $S$, i.e. a basis, and the canonical element $\Can_S(x)$.

\paragraph{Counting Subspaces.}
The number of $k$-dimensional subspaces of a $d$-dimensional subspace of $\bbF_2^n$ is given by the Gaussian binomial coefficient
\[
    \binom{d}{k}_2 = \prod_{i=0}^{k-1} \frac{2^{d-i} - 1}{2^{i+1} - 1}
\]
Similarly, the number of $d$-dimensional superspaces of a $k$-dimensional subspace $S$ is $\binom{n-k}{d-k}_2$, the number of $(d-k)$-dimensional subspaces of the quotient space $\bbF_2^n/S$

The Gaussian binomial coefficient can be upper bounded by
\begin{align*}
    \binom{d}{k} 
    &< \prod_{i=0}^{k-1} \frac{2^{d-i}}{2^{i}} 
    \\
    &= 2^{(k-1)d - k(k-1)}
    \\
    &= 2^{(k-1)(d-k)}
\end{align*}
An additional useful identity is
\[
    \frac{\binom{d-1}{k - 1}_2}{\binom{d}{k}_2}
    = \frac{2^{k}-1}{2^{d}-1}
\]
Thus, we can upper bound
\begin{align*}
    \frac{\binom{d-r}{k - r}_2}{\binom{d}{k}_2}
    &=
     \left(\frac{\binom{d-r}{k - r}_2}{\binom{d-(r-1)}{k-(r-1)}_2}\right) \cdot \left(\frac{\binom{d-(r-1)}{k - (r-1)}_2}{\binom{d-(r-2)}{k-(r-2)}_2}\right) \cdot \dots \left(\frac{\binom{d-1}{k-1}_2}{\binom{d}{k}_2}\right)
     \\
     &= \prod_{i=0}^{r-1} \frac{2^{k-i}-1}{2^{d-i}-1}
     \\
     &< 2^{-r(d-k)}
\end{align*}

\subsection{Gentle Measurement Lemma}

\begin{lemma}[Gentle Measurement Lemma \cite{DBLP:journals/tit/Winter99}]\label{lem:gentle-measurement}
Let $\rho$ be a quantum state in some Hilbert space, and let $\{\Pi, I-\Pi\}$ be a projective measurement that acts on that Hilbert space. Also, let $(\rho, \Pi)$ satisfy: $\Tr(\Pi \rho)\geq 1-\delta$. \par 

Next, let $\rho'$ be the state that results from applying $\{\Pi, I-\Pi\}$ to $\rho$ and post-selecting on obtaining the first outcome:
\[
\rho' = \frac{\Pi \rho \Pi }{\Tr(\Pi \rho)}
\]
Then $\mathsf{TraceDist}(\rho, \rho')\leq 2\sqrt{\delta}$.
\end{lemma}

\subsection{Oblivious State Preparation}

Oblivious state preparation allows a sender to instruct a receiver to prepare a quantum state, using only classical messages~\cite{AC:CCKW19,C:BK25}. It guarantees that the receiver successfully constructs the desired state, up to some Pauli corrections which the sender learns. Oblivious state preparation has been considered by prior works in the context of sending states from the set $\{\ket{0}, \ket{1},\ket{+},\ket{-}\}$. Here we define a slightly more general version for the preparation of \emph{coset states}
\[
    \ket{S_{x,z}} = X^xZ^z \sum_{s\in S} \ket{S}
\]

\begin{definition}
    An oblivious coset state preparation is an interactive protocol $\OSP$ between a PPT sender $\Sender$ with input a subspace $S\subset \bbF_2^n$ and a QPT receiver $\Recvr$:
    \[
        ((x,z),\  X^x Z^z \ket{S}) \gets \OSP\langle \Sender(S), \Recvr\rangle 
    \]
    where $(x, z) \in \{0,1\}^{2n}$ is the sender's output and $X^xZ^z\ket{S}$ is the receiver's output. It must satisfy the following properties:
    \begin{itemize}
        \item \textbf{Correctness:} With overwhelming probability, the receiver's output is close to the expected state $X^x Z^z \ket{S}$:
        \[
            \expect{  |\braket{S_{x,z} | \psi}|^2 : ((x,z), \ket{\psi}) \gets \OSP\langle \Sender(S), \Recvr\rangle } = 1-\negl
        \]
        \item \textbf{Security:} For any adversarial QPT receiver $\Recvr^*$ and any pair of subspaces $S_0, S_1 \subset \bbF_2^n$,
        \[
            \left| \Pr[b = 0 : ((x,z), b) \gets \langle \Sender(S_1), \Recvr^*\rangle] - \Pr[b = 0 : ((x,z), b) \gets \OSP\langle \Sender(S_1), \Recvr^*\rangle]\right| = \negl
        \]
    \end{itemize}
\end{definition}

Shmueli~\cite{STOC:Shmueli22} showed how to construct such a protocol using quantum hybrid fully homomorphic encryption.\footnote{Shmueli additionally requires subexponential security to derive an unclonability property. However, his protocol satisfies basic oblivious coset state preparation even with polynomial hardness.}

\begin{theorem}[\cite{STOC:Shmueli22}]
    Assuming quantum hybrid fully homomorphic encryption, there exists a protocol for oblivious coset state preparation.

    As a corollary, one exists assuming LWE.
\end{theorem}




\subsection{Certified Deletion}
\cite{EC:BGKMRR24} showed that given a coset state $\ket{S_{x.z}}$, if an adversary produces $v \in S+ x$, then $\langle v, \vec{1}\rangle$ becomes \emph{information-theoretically} hidden from the adversary's view. Moreover, the result holds even if the adversary additionally receives some auxiliary information that computationally (partially) hides $S$, $x$, and $\langle v, \vec{1}\rangle$. 
We will use a slightly different version which follows from some slight tweaks to their proof.

\begin{definition}{Certified Deletion Search Game.} The certified deletion game $\CDGame_{D_1, D_2}(\adv_1, \adv_2)$ is parameterized by two distributions $D_1$ and $D_2$ and is played by two adversaries $\adv_1$ and $\adv_2$. 
$D_1$ is a distribution over $(S + x, T, R^\perp + z_{R^\perp}, k)$ where $R\subset S \subset T \subset \bbF_2^\secpar$ are subspaces with dimensions $d_R$, $d_S$, and $d_T$, respectively, where $x\in \co(S)$, and where $z_{R^\perp} \in \co(R^\perp)$.
$D_2$ is a distribution over (potentially entangled) quantum states $(\aux_1, \aux_2)$. The game is played as follows.
\begin{enumerate}
    \item Sample $(S + x, T, R^\perp + z_{R^\perp}, k) \gets D_1$. Sample $z \gets \co(S^\perp) \intersect R^\perp + z_{R^\perp}$ and let $x_T = \Can_T(x)$.
    \item Sample $(\aux_1, \aux_2) \gets D_2(\ket{S_{x,z}}, S+x, T + x_T, R^\perp + z_{R^\perp}, k)$.
    \item Run $(v_1, \rho) \gets \adv_1(\aux)$.
    \item Run $v_2 \gets \adv_2(\rho, \aux_2)$.
    \item $\adv_1$ and $\adv_2$ win if $v_1 \in S+ x$ and $v_2\in S^\perp + z$.
\end{enumerate}
\end{definition}

There are three main changes to the game as compared to \cite{EC:BGKMRR24}'s original version. 
First, $\adv_2$'s goal is now to search for $v_2\in S^\perp + z$, rather than decide some bit. 
Second, $\ket{S_{x,z}}$ is no longer a uniformly random coset state. Instead, parts of its description are generated by a distribution $D_1$. In the modified theorem, we will essentially require that $D_1$ is \emph{pseudorandom}. This allows more flexible generation of coset states, such as generating their descriptions via PRF.
Third, the auxiliary input distribution may depend on $\ket{S_{x,z}}$. \cite{EC:BGKMRR24}'s version considered giving the adversary $(\ket{S_{x,z}}, \aux_1)$. By incorporating $\ket{S_{x,z}}$ into $\aux_1$, the auxiliary input can now be the leftover state from protocol executions involving $\ket{S_{x,z}}$. This allows the certified deletion theorem to be applied to compositions of protocols.
A final minor difference is that we swap the role of the computational and Hadamard basis so that the computational basis holds the certificate and the Hadamard basis holds the data.

\begin{theorem}\label{thm:search-cd}
    If there exists a QPT algorithm $\Sim$ such that
    \begin{gather*}
        \left\{
            \aux_1 : \begin{array}{c}
                (S+x, T, R^\perp + z_{R^\perp}, k) \gets D_1
                \\
                 (\aux_1, \aux_2) \gets D_2\big(\ket{S_{x,z}}, S+x, T + x_T, R^\perp + z_{R^\perp}, k \big)
            \end{array}
        \right\}
        \\
        \approx_c
        \\
        \left\{
            \aux_1 : \begin{array}{c}
                R, S, T \gets \Subspc(\bbF_2^\secpar, (d_R, d_S, d_T)) \text{ s.t. } R\subset S \subset T
                \\
                x \gets \co(S),\ z\gets \co(S^\perp),\ x_T = \Can_T(x),\ z_{R^\perp} = \Can_{R^\perp}(S)
                \\
                 (\aux_1, \aux_2) \gets \Sim\big(\ket{S_{x,z}}, T + x_T, R^\perp + z_{R^\perp}\big)
            \end{array}
        \right\}
    \end{gather*}
    and $D_1$ is supported on $(S, T, R)$ such that $2^{-(d_S - d_R)} = \negl$ and $2^{-(d_T - d_S)} = \negl$, then for all QPT $\adv_1$ and (unbounded) quantum $\adv_2$,
    \[
        \Pr[\mathsf{win}\gets \CDGame_{D_1, D_2}(\adv_1, \adv_2)] = \negl
    \]
\end{theorem}

The proof is almost identical to the one from \cite{EC:BGKMRR24}, with some small, but targeted, tweaks, so we defer it to \Cref{app:cd}. 
At a high level, the prior work purifies the generation of $\ket{S_{x,z}}$ by instead preparing the state
\[
    \ket{\psi_0} \propto \sum_{z \in (\co(S^\perp)\intersect R^\perp) + z_{R^\perp}} \ket{S_{x,z}}_{\calS} \otimes \ket{z}_{\mathsf{prep}}
\]
and giving the adversary register $\calS$.\footnote{In the original work, there is an additional hybrid before purification where the challenger guesses a message bit in order to make $z$ uniformly random (within $R^\perp + z_{R^\perp}$). We do not need this step in the search version, because $z$ starts as uniformly random (within $R^\perp + z_{R^\perp}$).} 
They then show that conditioned on $\adv_1$ producing $v_1 \in S + x$, the $\mathsf{prep}$ register holds $\sum_{z\in (\co(S^\perp)\intersect R^\perp) + z_{R^\perp}} p_z\ket{z}$, where $p_z$ is a phase with absolute value $1$, \emph{unentangled with the adversary's state}. Thus, measuring it gives a random $z$ from a set of size $2^{d_S - d_R}$ independently of $v_2$, which matches with negligible probability. 

The tweak occurs when arguing that the $\mathsf{prep}$ register holds $\sum_{z\in  (\co(S^\perp)\intersect R^\perp) + z_{R^\perp}} p_z \ket{z}$. Prior work used a computational property called subspace hiding to replace $\aux_1$ by an $\aux'_1$ that depends \emph{only} on (random) $T + x_T$ and $R^\perp + z_{R^\perp}$. We can achieve the same effect directly by using the simulation property of $D_1$ and $D_2$. Due to the gap between $T$ and $S$, $\adv_1(\aux'_1)$ cannot find an element of $S+x$ without fully measuring the coset state they are given, which collapses the $\mathsf{prep}$ register to the desired state. Since the condition on $\mathsf{prep}$ is efficiently testable, this must also occur in the purified experiment.

\subsection{Estimating Quantum Acceptance Probabilities}

\cite{TCC:Zhandry20} gives a method of approximating the probability that a state is accepted by a POVM $(\calP= \sum_{i} p_i (I- P_i), \calQ= \sum_{i} p_i (I- P_i))$ which is a  mixture of binary-outcome projective measurements $\{P_i, I-P_i\}$. Crucially, the method is almost-projective. In other words, if run twice, it will almost certainly give the same result both times. Later, \cite{C:ALLZZ21} observed that the technique can be applied to test if a state's acceptance probability is greater than some threshold.

Although the technique is quite general, we only need a few very specific properties from it that are given in corollary 1 of \cite{C:ALLZZ21}.
We refer the reader to \cite{TCC:Zhandry20} for a fully detailed description of the general technique.

Roughly, $\ATI_{\calP, \gamma}^{\epsilon,\delta}$ estimates whether a state $\rho$ is accepted by $\calP$ with probability at least $\gamma$, with accuracy $\delta$ and precision $\epsilon$.

\begin{lemma}
    \label{lem:ati}
    Let $(\calP = \sum_{i} p_i P_i, \calQ= \sum_{i} p_i (I- P_i))$ be a mixture of projective measurements such that it is efficient to sample from the distribution defined by $\Pr[i] = p_i$. For every $\epsilon, \gamma,\delta > 0$ and mixture of projective measurements $\calP$, there exists an algorithm $\ATI_{\calP,\gamma}^{\epsilon,\delta}$ outputting $\Accept$ or $\Reject$ such that the following hold.
    \begin{itemize}
        \item \textbf{Efficient.} The expected running time of $\ATI$ is $\poly[1/\epsilon, \log(1/\delta)]$,
        
        \item \textbf{Approximately Projective.} $\ATI$ is approximately projective. In other words, for all states $\rho$,
        \[
            \Pr\left[b_1 = b_2: \begin{array}{c}
                 (b_1, \rho') \gets \ATI(\rho)  \\
                 b_2 \gets \ATI(\rho') 
            \end{array}\right] = 1-\negl(\secpar)
        \]

        \item \textbf{Estimation.} For all states $\rho$,
        \[
            \Pr\left[
            \begin{array}{c}
                 b = \Accept \land \\
                 \Tr[\calP \rho'] < \gamma - \epsilon
            \end{array}
            :
             (b_1, \rho') \gets \ATI(\rho) \right] 
             \leq \delta
        \]
    \end{itemize}
\end{lemma} \fi
\section{Simultaneous Search}\label{sec:simult-search}

In this section, we prove a general lemma about asymetrically changing challenge distributions in no-signaling search tasks. 
Imagine that $B$ was given a register $\calB$ and $C$ was given a register $\calC$, then they each attempt to separately solve some search task, with the goal of succeeding simultaneously. 
Now consider what would happen if $C$ were instead given $\calC'$ which is indistinguishable to $C$, \emph{but not to $B$}. Although $C$ must still succeed in its search task (as long as it can check it) with almost the same probability, it is not a-priori clear that the simultaneous distribution is unaffected.
We show that $\calC$ can be replaced by $\calC'$ without dropping the simultaneous success probability to negligible.

\begin{lemma}\label{lem:simult-search}
    Let $\calD$ be a distribution outputting three registers $(\regA, \regB, \regC)$ along with a classical string $\cAuxClass$. Let $\calD'$ be a distribution taking as input $\cAuxClass$ outputting $\regC'$ such that
    \[
        \{(\regA, \regC, \cAuxClass): (\regA, \regB, \regC, \cAuxClass)\gets \calD\} 
        \approx_c
        \left\{(\regA, \regC', \cAuxClass): 
            \begin{array}{c}
             (\regA, \regB, \regC, \cAuxClass)\gets \calD \\
              \regC' \gets \calD'(\cAuxClass)
            \end{array}
            \right\}
    \]
    For all QPT algorithms $(A, C)$ and all (unbounded) quantum algorithms $B$, if 
    \[
        \Pr\left[1\gets B(\regB, \regA_B) \land 1\gets C(\regC, \cAuxClass, \regA_C): 
        \begin{array}{c}
             (\regA, \regB, \regC, \cAuxClass)\gets \calD \\
             (\regA_B, \regA_C) \gets A(\regA) 
        \end{array}
        \right] 
        \geq 1/p
    \]
    for some $p = \poly$, then there exists $q = \poly$ such that
    \[
        \Pr\left[
        1\gets B(\regB, \regA_B) \land 1\gets C(\regC', \cAuxClass, \regA_C): 
        \begin{array}{c}
             (\regA, \regB, \regC, \cAuxClass)\gets \calD \\
             (\regA_B, \regA_C) \gets A(\regA) \\
             \regC' \gets \calD'(\cAuxClass)
        \end{array}
        \right]
        \geq 1/q
    \]
\end{lemma}
\ifsubmission
    We defer the proof to \Cref{app:simult-search} due to page limits.
    At a high level, the proof begins by switching from viewing the indistinguishability of $\calD'$ as a hard decisional task to a hard search task: given many samples $\calC'$ from $\calD'(\cAuxClass)$ and one original sample $\calC$, find the original sample. Then, we decompose $B$ and $C$'s joint state using a projector onto $B$'s success:
    \[
        \alpha \ket{\phi^B_{B=1}} \otimes \ket{\psi^C_{B=1}} + \beta \ket{\phi^B_{B=0}} \otimes \ket{\psi^C_{B=0}}
    \] 
    Because of the orthogonality on $B$'s register and the fact that $B$ and $C$ cannot communicate, the branches $\ket{\psi^C_{B=1}}$ and $\ket{\psi^C_{B=0}}$ cannot interfere with one another. Thus, the reduction's overall success probability in the search game is at least as high as its success probability in the branch $\ket{\psi^C_{B=1}}$.
    Finally, we can find the correct $\calC$ in branch $\ket{\psi^C_{B=1}}$ if the lemma were false by sequentially testing $C$'s success on each candidate and outputting the first to succeed. If $C(\calC')$ could not succeed simultaneously with $B$, then $C$ will reject any false candidates $\calC'$ with overwhelming probability in this branch. The first candidate that it can accept is the real $\calC$.

\else
    \begin{proof}
    The precondition on $\calD$ and $\calD'$ implies a search hardness property. For any $p = \poly$, consider the task of finding one sample of $\regC$ when given it along with $p-1$ samples of $\calD'$, in random order, along with $\regA$:
    \[
        \left\{(\regA, \cAuxClass, \pi(\regC, \regC_1', \dots, \regC_{p-1}')): 
        \begin{array}{c}
            (\regA, \regC) \gets \calD
            \\
            \regC_i \gets \calD'\\
            \pi \gets \mathsf{Permute}_{p}
        \end{array}\right\}
    \]
    Any QPT algorithm has a $1/p + \negl$ probability of finding the index of the sample from $\calD$. This follows from a simple hybrid argument showing that any two indices of $\calD$ are indistinguishable. Both cases are indistinguishable from the experiment where every index contains a sample from $\calD'(\cAuxClass)$.\footnote{$\cAuxClass$ being classical ensures that an arbitrary number of samples can be generated.}

    Suppose that $(A,B,C)$ simultaneously output $1$ with probability $\geq 1/p$ for $p=\poly$. We show that if 
    \[
        \Pr\left[
        1\gets B(\regB, \regA_B) \land 1\gets C(\regC', \cAuxClass, \regA_C): 
        \begin{array}{c}
             (\regA, \regB, \regC, \cAuxClass)\gets \calD \\
             (\regA_B, \regA_C) \gets A(\regA) \\
             \regC' \gets \calD'(\cAuxClass)
        \end{array}
        \right]
        = \negl
    \]
    then $(A, B)$ can be used to violate the search hardness property over $2p$ candidates.
    The reduction is as follows.
    \begin{enumerate}
        \item The reduction takes as input $2p$ candidates $(\regC_1, \dots, \regC_{2p})$ along with $\regA$.
        \item The reduction runs $(\regA_B, \regA_C) \gets A(\regA)$.
        \item For each $i=1,\dots, 2p$, the reduction measures the output bit of $C(\regC_i, \cAuxClass, \regA_C)$. It outputs the first index where $C$ outputs $1$.
    \end{enumerate}

    In general, we may write $A$, $B$, and $C$ as projectors $(\Pi_A^0, \Pi_A^{1})$, etc, by extending their state with ancillas. For convenience, we consider these ancillas to be part of $\regA$, $\regA_B$, and $\regA_C$. After running $A$, the state across registers $(\regB, \regA_B, \regA_C)$ and any $\regC_i$ is a mixed state (potentially entangled with the original $\regC$)\footnote{We consider each candidate $\regC_i$ to have a copy of $\cAuxClass$, which is possible because $\cAuxClass$ is a classical string.}
    \[
        \rho_{\regB,\regA_B, \regA_C, \regC_i}
    \]
    The state on just register $\regA_C$ is
    \begin{align*}
        \Tr_{\regB,\regA_B}[\rho_{\regB,\regA_B, \regA_C, \regC_i}]
        &=
        \Tr_{\regB,\regA_B}[(\Pi_B^{0} \otimes I)\rho_{\regB,\regA_B, \regA_C, \regC_i}(\Pi_B^{0} \otimes I)] + \Tr_{\regB,\regA_B}[(\Pi_B^{1} \otimes I)\rho_{\regB,\regA_B, \regA_C, \regC_i}(\Pi_B^{1} \otimes I)]
        \\
        &\coloneqq \Tr_{\regB,\regA_B}[\rho^0_{\regB,\regA_B,\regA_C, \regC_i}] + \Tr_{\regB,\regA_B}[\rho^1_{\regB,\regA_B,\regA_C, \regC_i}]
    \end{align*}
    where we defined
    \[
        \rho^b_{\regB,\regA_B,\regA_C, \regC_i} 
        \coloneqq 
        (\Pi_B^{b} \otimes I)\rho_{\regB,\regA_B, \regA_C, \regC_i}(\Pi_B^{b} \otimes I)
    \]
    This holds because $(\Pi_B^b\otimes I)$ operates nontrivially only on the traced out registers.

    Now consider the final step of the reduction. Let $i^*$ be the index of the real sample $\regC$. The success or failure of the reduction is described by a POVM $(P, Q)$ where the result $P$ corresponds to a sequence of verifier decisions where the first index which it accepts is $i^*$. Thus, the reduction's success probability is
    \begin{align*}
        &\Tr\left[P\left(\Tr_{\regB,\regA_B}\left[\rho^0_{\regB,\regA_B,\regA_C, \regC_{1,\dots,2p}}\right] + \Tr_{\regB,\regA_B}\left[\rho^1_{\regB,\regA_B,\regA_C, \regC_{1,\dots,2p}}\right]\right)\right]\\
        &= \Tr\left[P\Tr_{\regB,\regA_B}\left[\rho^0_{\regB,\regA_B,\regA_C, \regC_{1,\dots,2p}}\right]\right] + \Tr\left[P\Tr_{\regB,\regA_B}\left[\rho^1_{\regB,\regA_B,\regA_C, \regC_{1,\dots,2p}}\right]\right]
        \\
        &\geq \Tr\left[P\Tr_{\regB,\regA_B}\left[\rho^1_{\regB,\regA_B,\regA_C, \regC_{1,\dots,2p}}\right]\right]
    \end{align*}
    Consider testing an index $i<i^*$. For the first such $i$, the probability of accepting (and thus the reduction guessing $i$) when testing $\rho^1_{\regB,\regA_B,\regA_C, \regC_{1,\dots,2p}}$ is
    \begin{align*}
        &\Tr\left[ \left(I\otimes \Pi_C^1 \right) \Tr_{\regB,\regA_B}\left[\rho^1_{\regB,\regA_B,\regA_C, \regC_i}\right] \left(I\otimes \Pi_C^1\right)\right]
        \\
        &= \Tr\left[\Tr_{\regB,\regA_B}\left[\left(\Pi_B^1 \otimes \Pi_C^1\right)\rho_{\regB,\regA_B,\regA_C, \regC_i} \left(\Pi_B^1 \otimes \Pi_C^1\right)\right]\right]
        \\
        &= \negl
    \end{align*}
    by assumption. Thus, it is a gentle measurement and applying it disturbs the state negligibly (\Cref{lem:gentle-measurement}). Inducting on $i< i^*$ shows that the probability of the reduction guessing any $i<i^*$ is negligible and the state when testing $\regC_{i^*}$ is negligibly far from the state prior to any tests. Therefore the probability of $i^*$ being the first accepting index is negligibly far from
    \begin{align*}
        \Tr\left[ \left(I\otimes \Pi_C^1\right) \Tr_{\regB,\regA_B}\left[\rho^1_{\regB,\regA_B,\regA_C, \regC_i}\right] \left(I\otimes \Pi_C^1\right)\right]
        &= \Tr\left[\Tr_{\regB,\regA_B}\left[\left(\Pi_B^1 \otimes \Pi_C^1\right)\rho_{\regB,\regA_B,\regA_C, \regC_i} \left(\Pi_B^1 \otimes \Pi_C^1\right)\right]\right]
        \\
        &= 1/p
    \end{align*}
    Since $1/p > 1/(2p) + \negl$, the reduction contradicts the search-security of $\calD$ and $\calD'$.
    \ifsubmission \qed \fi
\end{proof}
\fi
\section{Proofs of No-Intrusion for Coset States}\label{sec:poni-coset}

In this section, we build a method of non-destructively testing a coset state $\ket{S_{x,z}}$, using only classical communication. Informally, the protocol will act as a proof of knowledge of a vector in $S+x$. 
By combining this with certified deletion, no adversary who previously tried to copy $\ket{S_{x,z}}$ can simultaneously find a vector in $S^\perp + z$.

\begin{construction}[PoNI for Coset States] 
    \label{constr:coset-poni}
    Let $S\subset \bbF_2^\lambda$ be a $\dimS$-dimensional subspace and let $x\in \co(S)$ and $z\in \co(S^\perp)$. Let $T \gets \Subspc(S, \dimT)$ be a random superspace of $S$ and let $x_T = \Can_T(x)$.
    At the beginning of the protocol, the prover holds $\ket{S_{x,z}}$ in register $\calR_S$ and the verifier holds $(T,x_T)$.
    \begin{enumerate}
        \item \textbf{Verifier:} Sample a random $3\secpar/4$-dimensional superspace $T\supset S$.
        \item Perform an oblivious state preparation for $\ket{T}$, where the verifier acts as $\Sender$ and the prover acts as $\Recvr$.
        \[
            ((x',z'), X^{x'}Z^{z'}\ket{T}) \gets \OSP\langle(\Sender(T), \Recvr\rangle
        \]
        At the end of this, the verifier holds $(x',z')$ and the prover holds $\ket{T_{x',z'}}$ in register $\calR_T$ (up to global phase).\footnote{In a 2-round oblivious state preparation for coset states where the final message is sent by the receiver, the final message can be combined with the next step.}
        
        \item \textbf{Prover:} Perform a CNOT operation from $\calR_{S}$ to $\calR_T$ to obtain the state
        \begin{equation}\label{eq:coset-poni-middle}
            \ket{S_{x,z-z'}}^{\calR_S} \otimes \ket{T_{x+x',z'}}^{\calR_T}
        \end{equation}
        Measure register $\calR_T$ in the computational basis to obtain a vector $v\in T + x + x'$. Send $v$ to the verifier. Output the residual state on register $\calR_S$.

        \item \textbf{Verifier:} If  $v \in T+x_T+x'$, output $(\Accept, (0,z'))$. Otherwise output $\bot$.
    \end{enumerate}
\end{construction}

\paragraph{Security Game.}
To define security of the scheme, we consider the following security game $\PoNIMoE_{D_1, D_2}(\adv_1, \adv_P, \adv_H)$ which is parameterized by a pair of distributions $(D_1, D_2)$ and played by a triple of adversaries $(\adv_1, \adv_P, \adv_H)$. $D_1$ is supported on $(S+x,T, R^\perp + z_{R^\perp}, k)$ where $R\subset S\subset T\subset \bbF_2^\secpar$ are subspaces with dimensions $d_S$, $d_T$, and $d_R$, respectively, where $x\in \co(S)$, and where $z_{R^\perp}\in \co(R^\perp)$. The game is played as follows
\begin{enumerate}
    \item Sample $(S + x, T, R^\perp + z_{R^\perp}, k) \gets D_1$. Sample $z \gets (\co(S^\perp) \intersect R^\perp) + z_{R^\perp}$ and let $x_T = \Can_T(x)$.
    \item Sample $(\aux_1, \aux_H) \gets D_2(\ket{S_{x,z}}, S+x, T + x_T, R^\perp + z_{R^\perp}, k)$.
    \item Run $(\calP, \calH) \gets \adv_1(\aux_1)$.
    \item Run $b\gets \PoNI\langle\adv_P(\calP), V(T,x_T)\rangle$ where $b$ is the verifier's decision bit.
    \item Run $v_H \gets \adv_H(\regH, \aux_H)$.
    \item The adversary wins if $b = \Accept$ and $v_H \in S^\perp + z$.
\end{enumerate}
Roughly, the security game captures the probability that $\adv_P$ can pass a PoNI simultaneously with $\adv_H$ having stolen $v_H \in S^\perp + z$. 

\paragraph{Composability.}
The auxiliary input $\aux_1$ controls the self-composability of the protocol and what side information about $\ket{S_{x.z}}$ can be safely released. One can think of it as the residual state of prior $\PoNI$ executions together with any additional side-information about $\ket{S_{x,z}}$, such as a ciphertext $\Enc(S)$. 
During the process of the theft, represented by $\adv_1$, $\adv_1$ may use $\aux_1$ to inform how it steals $v_H$ undetectably. 
Later, $\adv_H$ may separately steal additional information $\aux_H$ -- for example, the decryption key for $\Enc(S)$. The PoNI must be secure even with these two sources of side-information.

Together, $D_1$ and $D_2$ characterize $(\aux_1, \aux_2)$ and the distribution of $\ket{S_{x,z}}$.
Because $D_2$ depends on both $\ket{S_{x,z}}$ and $T+ x_T$, it may include the prover's view in prior $\PoNI$ executions. 
$D_1$ controls the distribution of $\ket{S_{x,z}}$, along with the verification key $T+x_T$. For example, it may be a PRF evaluation using key $k$.

To encode self-composability, we prove two properties of our PoNI. First, $T+x_T$ and $S+x$ are hidden from the prover's view in an execution. Second, the PoNI is secure for any auxiliary information $\aux_1$ which semantically hides $T+x_T$ and $S+x$. 
Putting these together, the PoNI is secure even if $\aux_1$ is the leftover state from the prover's view in a prior execution.

\begin{theorem}[PoNI for Cosets]\label{thm:poni-coset}
    Let $\PoNI$ be the protocol from \Cref{constr:coset-poni}. It satisfies the following properties.
    \begin{itemize}
        \item \textbf{Correctness and State Preservation up to Z Error.} For every $\dimS$-dimensional $S$, every $\dimT$-dimensional $T\supset S$, and every $x,z\in \bbF_2^\secpar$ with $x_T = \Can_T(x)$, at the end of an execution
        \[
            (\rho, (b, (0, z'))) \gets \PoNI\langle P(\ket{S_{x,z}}), V(T,x_T)\rangle
        \]
        where the prover outputs $\rho$ and the verifier outputs $(b, (0,z'))$, it is the case that $b= \Accept$ and $\rho = Z^{z'}\ketbra{S_{x,z}}Z^{z'}$.
    
        \item \textbf{Semantic Security.} Fix any $(S,x,z)$ along with $T\supset S$ and $T'$. Any QPT adversarial prover's view is indistinguishable between $\PoNI\langle\adv_P(\ket{S_{x,z}}, V(T, x_T)\rangle$ and $\PoNI\langle\adv_P(\ket{S_{x,z}}, V(T', 0)\rangle$.
        
        \item \textbf{Search Monogamy of Entanglement.} 
        Let $d_R = \dimR$, $d_S = \dimS$, and $d_T = \dimT$. If there exists a QPT algorithm $\Sim$ such that
        \begin{gather*}
        \left\{
            (\aux_1,\ S,\ x,\ T,\ R) : 
            \begin{array}{c}
                (S+x, T, R^\perp + z_{R^\perp}, k) \gets D_1
                \\
                z \gets (\co(S^\perp)\intersect R^\perp) + z_{R^\perp},\ x_T = \Can_T(x)
                \\
                 (\aux_1, \aux_2) \gets D_2\big(\ket{S_{x,z}}, S+x, T + x_T, R^\perp + z_{R^\perp}, k \big)
            \end{array}
        \right\}
        \\
        \approx_c
        \\
        \left\{
            (\aux_1,\ S,\ x,\ T,\ R) : 
            \begin{array}{c}
                S\gets \Subspc(\bbF_2^\secpar, \dimS)
                \\
                x \gets \co(S),\ z\gets \co(S^\perp)
                \\
                 (\aux_1, \aux_2) \gets \Sim\big(\ket{S_{x,z}}\big)
                 \\
                 T \gets \Supspc(S, \dimT),\ x_T = \Can_T(x)
                 \\
                 R \gets \Subspc(S,\dimR)
            \end{array}
        \right\}
    \end{gather*}
    then for all QPT $(\adv_1, \adv_P)$ and (unbounded) quantum $\adv_H$,
    \[
        \Pr[\mathsf{win} \gets \PoNIMoE_{D_1, D_2}(\adv_1, \adv_P, \adv_H)] = \negl
    \]
        
    \end{itemize}
\end{theorem}
\begin{proof}
    Semantic security follows immediately from the security of $\OSP$.

    Correctness and state preservation follow from observation if \cref{eq:coset-poni-middle} holds. To see that it holds, observe that the result of the CNOT operation is
    \begin{align*}
        &\sum_{s\in S} (-1)^{s\cdot z} \ket{s+x} \otimes \sum_{t\in T} (-1)^{t\cdot z'} \ket{t+x'+s+x}
        \\
        &= \sum_{s\in S} (-1)^{s\cdot z} \ket{s+x} \otimes \sum_{t'\in T} (-1)^{(t' - s)\cdot z'} \ket{t'+x'+x}
        \\
        &= \ket{S_{x,z-z'}} \otimes \ket{T_{x+x',z'}}
    \end{align*}
    by doing a change of variables $t' \coloneqq t+s$ using the fact that $S\subset T$.

    We prove search monogamy of entanglement in \Cref{subsec:poni-coset:security}.
    \ifsubmission \qed \fi
\end{proof}

\paragraph{Overview of the Proof of Search MoE.}
The remainder of this section is dedicated to proving search monogamy of entanglement. As an overview, we will try to reduce to the coset certified deletion theorem by showing how to extract a vector $v_1 \in S+x$ from any $\adv_P$ which gives an accepting PoNI with noticeable probability. The certified deletion theorem gives the extractor the following information for free: a subspace $R\subset S$, a superspace $T\supset S$, and a vector $x_T \in T + x$.

Given these, it needs to find a vector in $S+ x$ by just running $\PoNI$ executions with $\adv_P$. Moreover, it must do so \emph{simultaneously} with $\adv_H$ succeeding in finding a $v_H \in S^\perp + z$ in order to contradict certified deletion.

\Cref{subsec:coset-remainder} builds a tool for extracting a vector from $S+x$ using the ``free'' vector $x_T \in T + x$ and a vector from some other $(T' + S) + x$ where $(T\intersect T') \subset S$.
This reduces the extractor's task to just finding a vector from $(T' + S) + x$.

An honest execution of $\PoNI$ using $T$ results in the verifier learning a vector in $T+x$. So, by using a different $T'$, we could hope to find the other required vector. However, it is not simple to find an honest $T'\supset S$, so the extractor will have to resort to dishonest behavior. Unfortunately, $\adv_P$ does not have to answer a dishonest $\PoNI$ execution honestly -- even if the dishonest $\PoNI$ looked indistinguishable, they could implicitly abort, since the decision bit is hidden from their view.

\Cref{subsec:extractor} describes how the extractor can sample an appropriate $T'$ using the ``free'' information -- specifically the subspace $R$. Then, it builds several tools useful for arguing that $T'$ is indistinguishable from $T$, even given the other side information available. Thus, we can argue that $\adv_P$ would answer a $\PoNI$ using $T'$ \emph{without} relying on any security properties of $\PoNI$. Since the decision bit is accessible in this case, it cannot change between using $T$ and $T'$.

Finally, we combine everything to prove security in \Cref{subsec:poni-coset:security}. This makes use of the extractor built up in \Cref{subsec:coset-remainder,subsec:extractor}, as well as the simultaneous search lemma from \Cref{sec:simult-search}.

\subsection{Coset Remainder Lemma}\label{subsec:coset-remainder}

\begin{lemma}\label{lem:coset-remainder}
    Let $T_1$ and $T_2$ be subspaces of $\bbF^n$ and let $S\coloneqq T_1 \intersect T_2$.
    There exists a polynomial-time algorithm which for all inputs $(T_1, T_2, x_1, x_2)$ satisfying
    \begin{align*}
        x_1 &\in T_1 + x
        \\
        x_2 &\in T_2 + x
    \end{align*}
    for some $x\in \bbF^n$, the algorithm outputs a description of $S+x$.
\end{lemma}
\ifsubmission
    Due to page limits, we defer the proof to \Cref{app:coset-remainder}.
\else
    \begin{proof}
    Since $x_1 \in T_1 + x$ and $x_2 \in T_2 + x$, there exist $t_1 \in T_1$ and $t_2 \in T_2$ such that
    \begin{align*}
        x_1 &= t_1 + x
        \\
        x_2 &= t_2 + x
    \end{align*}
    Let $C_b$ be a complementary subspace to $S$ within $T_b$, i.e. $T_b = \Span(S, C_b)$ and $S\intersect C_b = \{\vec{0}\}$. Then we can uniquely decompose
    \begin{align*}
        t_1 &= s_1 + c_1
        \\
        t_2 &= s_2 + c_2
    \end{align*}
    where $s_1, s_2 \in S$, where $c_1 \in C_1$, and $c_2 \in C_2$. Using this decomposition,
    \[
        x_1 - x_2 = t_1 - t_2 = (s_1 - s_2) + (c_1 - c_2)
    \]
    Note that $c_1 - c_2 \in \Span(C_1, C_2)$. Since $S$ and $\Span(C_1, C_2)$ are complementary with respect to $\Span(T_1, T_2)$, it is efficient to find $(s_1 - s_2)$ and $(c_1 - c_2)$ via basis decomposition. Furthermore, $C_1 \intersect C_2 = \{0\}$ since $T_1\intersect T_2 = S$ and $S\intersect C_1 = \{0\}$. Therefore we can also efficiently decompose $(c_1 - c_2)$ into $c_1$ and $-c_2$. Finally, observe that
    \begin{align*}
        x_1 - c_1 &= s_1 + x \in S+x
        \\
        x_2- c_2 &= s_2 + x \in S+x
    \end{align*}

    In summary, the algorithm does the following:
    \begin{enumerate}
        \item Compute $S = T_1 \intersect T_2$.
        \item Pick $C_b$ such that $C_b$ and $S$ are complementary with respect to $T_b$, for $b \in \{1,2\}$. This can be done by starting with a basis for $S$, then adding random vectors from $T_b$ subject to linear independence and taking their span to be $C_b$.
        \item Compute $x_1 - x_2$, then decompose it into a summation of $s\in S$, $c_1 \in C_1$, and $-c_2 \in C_2$ using basis decomposition.
        \item Output a description of $S+(x_1 -c_1)$.
    \end{enumerate}
    \ifsubmission \qed \fi
\end{proof}
\fi

\begin{corollary} \label{coro:coset-remainder}
    Let $S$ be a subspace and let $T_1\supset S$ be a superspace of $S$. Let $T_2$ be a subspace such that $T_1 \intersect T_1 \subseteq S$. Let $x\in \bbF_2^n$, and let
    \begin{align*}
        x_1 &\in T_1 + x
        \\
        x_2 &\in T_2 + S + x
    \end{align*}
    On input $(x_1, x_2, T_1, T_2)$, the algorithm from \cref{lem:coset-remainder} outputs an element of $S+x$.
\end{corollary}
\begin{proof}
    Since $x_2 \in T_2 + S + x$, there exist $t_2 \in T_2$ and $s \in S$ such that 
    \[
        x_2 = t_2 + s + x
    \]
    In other words, $x_2 \in T_2 + (s + x)$.
    Furthermore, since $s \in S \subset T_1$, we have $T_1 + x = T_1 + (s+x)$, so $x_1 \in T_1 + (s + x)$. Therefore \cref{lem:coset-remainder} guarantees that the algorithm outputs an element of $(T_1 \intersect T_2) + (s + x)$. Since $s \in S$ and $(T_1 \intersect T_2) \subset S$, it must be the case that $(T_1 \intersect T_2) + (s + x) \subset S + x$.
    \ifsubmission \qed \fi
\end{proof}

\subsection{Extractor}\label{subsec:extractor}

Our ultimate goal is to reduce to certified deletion by building an extractor that uses $R\subset S$, $T\supset S$, and $x_T \in S + x$ to extract a vector $v_1 \in S+x$ from $\adv_P$. In this subsection, we describe the extractor and build several technical claims which will be useful for analyzing its behavior later.

\paragraph{Extractor Description.} The extractor takes as input $(R, T, x_T)$, along with an adversary $\adv_P(\regP)$, then does the following.
\begin{enumerate}
    \item Sample a subspace $T_R \gets \Supspc(R, \dimT)$,
    
    \item Run a $\PoNI$ with $\adv_P(\regP)$ using $T_R$, except for the decision step. Explicitly, execute an oblivious state preparation $\OSP\langle\Sender(T_R ), \adv_P(\regP)\rangle$ while acting as the sender. 
    Let $x_{\OSP}$ be the $X$ error incurred during the $\OSP$.\footnote{Knowledge of $x_{\OSP}$ requires seeing the sender's view of the $\OSP$, which removes any security properties it might have.} Receive a vector $v$ from $\adv_P$. which is masked by some $x_{\OSP}$ that is recovered from the $\OSP$ execution inside the $\PoNI$.

    \item Run the coset remainder algorithm (\cref{lem:coset-remainder}) on input $(T, T_R, x_T, v - x_{\OSP})$. Output the result.
\end{enumerate}

The main challenge is to show that $\adv_P$ will actually answer a $\PoNI$ using $T_R$ with a valid response. Since $T_R$ is not distributed according to the honest distribution, it is not clear that an adversary who directly saw $T'$ would have to answer. Furthermore, decoding the result of the PoNI requires the equivalent of decrypting a ciphertext that depends on $T_R$. Thus, we cannot hope to hide $T_R$ from $\adv_P$ when reasoning about whether their decision bit is valid or not. \Cref{claim:PoNI-coset-extractor-indist-with-coset-state} shows a scenario in which $\adv_P$ cannot distinguish between a correctly distributed $T$ and an incorrectly distributed one $T_W$. Then, \Cref{claim:poni-coset-extractor-1sample} shows that the extractor actually produces a single sample from that indistinguishable distribution, up to negligible error.

\begin{lemma}\label{claim:PoNI-coset-extractor-indist-with-coset-state}
    Fix any dimensions $d_S < d_T < d_W \in [0,n]$, any $d_W$-dimensional subspace $W\subset \bbF_2^n$, and any $x_W \in \co(W)$. Then
    \[
        \left\{(W, T, \ket{S_{x,z}}): \begin{array}{c}
            T\gets \Subspc(W,d_T) \\
             S\gets \Subspc(T, d_S) \\
             x \gets W + x_W,\ z \gets \bbF_2^n
        \end{array}\right\}
        =
        \left\{(W, T', \ket{S_{x,z}}): \begin{array}{c}
             S\gets \Subspc(W, d_S) \\
             x \gets W + x_W,\ z \gets \bbF_2^n \\
             T'\gets \Subspc(W,d_T)
        \end{array}\right\}
    \]
\end{lemma}
\begin{proof}
    For any $S$, the mixed state over possible $\ket{S_{x,z}}$ is
    \begin{align*}
        \propto \sum_{\substack{x\in W + x_W \\ z\in \bbF_2^n}} \ketbra{S_{x,z}}
        &= \sum_{\substack{x\in W + x_W \\ z\in \bbF_2^n}} Z^z \ketbra{S+x} Z^z
        \\
        &\propto \sum_{x\in W + x_W} \sum_{x' \in S + x} \ketbra{x'}
        \\
        &\propto \sum_{x\in W + x_W} \ketbra{x'}
    \end{align*}
    In other words, it is independent of $S$. Therefore the two distributions are both
    \[
        \left\{\left(W, T, \rho_{W,x_W}  \right): T\gets \Subspc(W,d_T) \right\}
    \]
    where $\rho_{W,x_W} \propto \sum_{x\in W + x_W} \ketbra{x'}$.
    \ifsubmission \qed \fi
\end{proof}

\begin{lemma}\label{claim:poni-coset-extractor-1sample}
    Fix a $d_S$-dimensional subspace $S\subset \bbF_2^\secpar$.
    Let $d_T$ and $d_R < \min(d_S, d_T)$ be positive integers and define $d_W = d_T + d_S - d_R$. 
    For any $d_S$, $d_T$, $d_R$ such that $\secpar - d_W = \omega(\log(\secpar))$ and $d_R = \omega(\log(\secpar))$,
    the following two distributions are statistically close:
    \[
        \left\{(W, T_W): \begin{array}{c}
            W \gets \Supspc(S, d_W) \\
            T_W\gets \Subspc(W, d_T)
        \end{array}\right\}
        \approx_{s}
        \left\{(T_R+S, T_R): 
        \begin{array}{c}
             R\gets \Subspc(S, d_R) \\
             T_R \gets \Supspc(R, d_T)
        \end{array}\right\}
    \]
\end{lemma}
\begin{proof}
    Consider the hybrid distribution
    \[
        \left\{(T_R+S, T_W): 
        \begin{array}{c}
             R\gets \Subspc(S, d_R) \\
             T_R \gets \Supspc(R, d_T) \\
             T_W \gets \Subspc(T_R + S, d_T)
        \end{array}\right\}
    \]
    We show that this distribution is statistically close to both distributions, beginning with the left-hand distribution $(W,  T_W)$.
    Once $R$ is fixed, $T_R$ can be sampled iteratively by defining $T'_{d_R} =R$, then sampling $T'_i$ from $i = d_R$ to $d_T$ by sampling $v_{i+1} \gets \bbF_2^\secpar\backslash T'_{d_R}$ and setting $T'_{i+1} = T'_i + \Span(v_{i+1})$. In this procedure, $T'_{d_T} = T_R$.
    Since $\secpar-d_T \geq \secpar-d_W = \omega(\log(\secpar))$, every $T'_i$ contains a negligible fraction of $\bbF_2^\secpar$. Therefore this distribution is statistically close to if we had sampled each $v_i \gets \bbF_2^\secpar$. 
     Furthermore, $T'_i + S$ also contains a negligible fraction of $\bbF_2^\secpar$ since $\secpar - (d_S + d_T - d_R) = \secpar - d_W = \omega(\log(\secpar))$. Therefore the distribution of $T'_i + S$ is statistically close to sampling $v_i \gets \bbF_2^\secpar\backslash (S + T_i)$ and adjoining $v_i$ to $S$. This is precisely a method of sampling a uniform $W \gets \Supspc(S, d_W)$, so
     \[
        \left\{(W, T_W): \begin{array}{c}
            W \gets \Supspc(S, d_W) \\
            T_W\gets \Subspc(W, d_T)
        \end{array}\right\}
        \approx_s
        \left\{(T_R+S, T_W): 
        \begin{array}{c}
             R\gets \Subspc(S, d_R) \\
             T_R \gets \Supspc(R, d_T) \\
             T_W \gets \Subspc(T_R + S, d_T)
        \end{array}\right\}
     \]

     Next, we show that the distributions of $T_R$ and $T_W$ are statistically close conditioned on $W = T_R + S$. Since $T_R$ is uniformly distributed over superspaces of $R$, its distribution conditioned on $T_R \subset W$ for any $W$ (in this case, $W = T_R + S$) is uniform over the intersection of $T_R$'s original support and the support of the event $T_R \subset W$. In other words, $T_R$ is a random subspace of $W$ subject to $T_R\supset R$ and $W = T_R + S$.

     Now consider the distribution of $T_W$ conditioned on $W$. It suffices to show that with overwhelming probability, $W = T_W + S$. Conditioned on this occurring, $T_W \intersect S$ is a uniformly random subspace of $S$ with dimension $d_T + d_S - d_W = d_R$, so $T_W$ is a uniform superspace of a random $R\subset S$ subject to $W = T_W + S$, the same distribution as $T_R$.

     Consider the event that $W \neq T_W + S$. This event can be decomposed into the (disjoint) events $\dim(T_W + S) = k$ for $k \in [d_R + 1, d_S]$. The number of possible intersections of dimension $k$ is $\binom{d_S}{k}_2 < 2^{k(d_S - k)}$. For each such $k$-dimensional subspace $K$, the probability of sampling a $T$ such that $T\supset K$ is given by counting the number of $d_T$-dimensional superspaces of $k$ within $W$ and dividing by the number of overall $d_T$-dimensional subspaces within $W$:
     \[
        \Pr_{T}[T\supset K] = \frac{\binom{d_W - k}{d_T -  k}_2}{\binom{d_W}{d_T}} < 2^{-k(d_W - d_T)}
     \]
     using the bounds from \Cref{sec:prelims-algebra}.
     By union bounding over all $\leq 2^{(k-1)(d_S - k)}$ subspaces of dimension $k \geq d_R + 1$,
     \begin{align*}
        \Pr[\dim(T_W \intersect S) = k] 
        &< 2^{k(d_S - k)} 2^{-k(d_W - d_T)}
        \\
        &= 2^{k(d_S + d_T - d_W - k)}
        \\
        &< 2^{-(d_R + 1)(d_R - (d_R + 1))}
        \\
        &= 2^{-\omega(\log(\secpar))} = \negl
     \end{align*}
     Union bounding over all $k \in [d_R + 1, d_S]$, the probability that $\dim(T_W + S) < d_W$ is negligible.

     Therefore
     \[
        \left\{(T_R + S, T_W): \begin{array}{c}
            R\gets \Subspc(S, d_R) \\
             T_R \gets \Supspc(R, d_T) \\
             T_W \gets \Subspc(T_R + S, d_T)
        \end{array}\right\}
        \approx_s
        \left\{(T_R+S, T_R): 
        \begin{array}{c}
             R\gets \Subspc(S, d_R) \\
             T_R \gets \Supspc(R, d_T)
        \end{array}\right\}
     \]
     and the claim holds by triangle inequality.
     \ifsubmission \qed \fi
 \end{proof}

\ifsubmission
\else
\subsection{Proof of Security}\label{subsec:poni-coset:security}
\begin{proof}[Proof of Property 3 from \Cref{thm:poni-coset}]
    We reduce to the search-style certified deletion theorem for coset states (\Cref{thm:search-cd}). 

    We first define the distributions $D'_1$ and $D'_2$ for the certified deletion theorem and show that they satisfy its precondition. 
    $D'_1$ is simply $D_1$. $D'_2(\ket{S_{x.z}}, S+x, T+x_T, R^\perp + z_{R^\perp}, k)$ is sampled by sampling $(\aux_1, \aux_H) \gets D_2(\ket{S_{x.z}}, S+x, T+x_T, R^\perp + z_{R^\perp}, k)$, setting $\aux'_1 = (\aux_1, T+x_T, R)$, and outputting $(\aux'_1, \aux_H)$.
    By the precondition of \Cref{thm:poni-coset}, there exists a QPT algorithm $\Sim$ such that
    \begin{gather} 
        \nonumber
        \left\{
            (\aux_1,\ T+x_T,\ R) : \begin{array}{c}
                (S+x, T, R^\perp + z_{R^\perp}, k) \gets D_1
                \\
                z \gets (\co(S^\perp)\intersect R^\perp) + z_{R^\perp},\ x_T = \Can_T(x)
                \\
                 (\aux_1, \aux_2) \gets D_2\big(\ket{S_{x,z}}, S+x, T + x_T, R^\perp + z_{R^\perp}, k \big)
            \end{array}
        \right\}
        \\
        \approx_c  \label{eq:poni-coset-precondition}
        \\
        \nonumber
        \left\{
            (\aux_1,\ T+x_T,\ R) : \begin{array}{c}
                S\gets \Subspc(\bbF_2^\secpar, \dimS)
                \\
                x \gets \co(S),\ z\gets \co(S^\perp)
                \\
                 (\aux_1, \aux_2) \gets \Sim\big(\ket{S_{x,z}}\big)
                 \\
                 T \gets \Supspc(S, \dimT),\ x_T = \Can_T(x)
                 \\
                 R \gets \Subspc(S,\dimR)
            \end{array}
        \right\}
    \end{gather}
    Define $\Sim'(\ket{S_{x,z}}, T + x_T, R^\perp + z_{R^\perp})$ to run $\aux_1 \gets \Sim(\ket{S_{x,z}})$, then output $(\aux_1, T+x_T, R)$. When restated in terms of $D'_2$ and $\Sim'$, \Cref{eq:poni-coset-precondition} is exactly the precondition for the coset certified deletion theorem.
    
    Suppose that some adversary $(\adv_1, \adv_P, \adv_H)$ won $\PoNIMoE$ with noticeable probability. We construct an adversary $(\adv_1', \adv_2')$ for certified deletion as follows.

    $\adv_1'$ takes as input $\aux'_1 = (\aux_1, T+x_T, R)$ and does the following. 
    \begin{enumerate}
        \item It runs $(\regP, \regH) \gets \adv_1(\aux_1)$. 
        \item It samples a subspace $T_R \gets \Supspc(R, \dimT)$.
        \item It runs $\PoNI\langle \adv_P(\regP), V(T_R, 0)\rangle$, acting as $V$, except for the decision step.
        
        Explicitly, execute an oblivious state preparation $\OSP\langle\Sender(T_R), \adv_P(\calP)\rangle$ while acting as the sender. During the course of this, it obtains an $X$ error $x_{\OSP}$. Afterwards, it receives a vector $v_P$ from $\adv_P$.
        
        \item It runs the coset remainder algorithm (\Cref{lem:coset-remainder}) on input $(T, T_R, x_T, v_P-x_{\OSP})$ to get result $v_1$.
        \item It outputs $(v_1, \calH)$.
    \end{enumerate}
    $\adv_2'$ takes as input $(\calH, \aux_H)$ and outputs $v_2 \gets \adv_H(\regH, \aux_H)$.

    We now show that $(\adv_1', \adv_2')$ simultaneously output $v_1 \in S+ x$ and $v_2 \in S^\perp + z$ with noticeable probability, violating \Cref{thm:search-cd}. 
    By \Cref{coro:coset-remainder}, $v_1 \in S + x$ whenever $v_P - x_{\OSP} \in (T_R + S) + x$ and $(T_R \intersect T) \subset S$.
    \Cref{subclaim:poni-sec-vp-in-tp-s-x} shows that there is a function $q = \poly$ such that
    \[
        \Pr_{\Hyb_1}\left[\big((v_P - x_{\OSP}) \in (T_R + S) + x\big)  \land \big(v_2 \in S^\perp + z\big)\right] \geq 1/q
    \]
    \Cref{subclaim:coset-poni-tp-t-intersection} shows that 
    \[
        \Pr[(T_R \intersect T) = R] = 1-\negl
    \]
    Combining these together with \Cref{coro:coset-remainder},
    \begin{align*}
        &\Pr[(v_1 \in S + x) \land (v_2 \in S^\perp + z)]
        \\
        &\geq \Pr\left[\bigg(\big((v_P - x_{\OSP}) \in (T_R + S) + x\big) \land (v_2 \in S^\perp + z)\bigg) \land \bigg((T\intersect T_R)  = R\bigg)\right]
        \\
        &\geq 1/q - \negl
    \end{align*}
    for some $q = \poly$, which is noticeable.

    The remainder of this proof is dedicated to proving Claims \ref{subclaim:poni-sec-vp-in-tp-s-x} and \ref{subclaim:coset-poni-tp-t-intersection}. 

    \begin{claim}\label{subclaim:coset-poni-tp-t-intersection}
        \[
            \Pr[(T_R \intersect T) = R] = 1- \negl
        \]
    \end{claim}
    \begin{proof}
        Fix a complementary space $V$ to $R$, i.e. $V+R = \bbF_2^\secpar$ and $V\intersect R = \{0\}$. We can write every $T\supset S$ as $T = R + A$ and $T_R$ as $T_R = R + A_R$ for subspaces $A, A_R \subset V$. Furthermore, $T\intersect T_R = R + (A \intersect A_R)$. 
        Fix $A$ and consider sampling $A_R$. For any vector $a \in A$, the number of possible $A_R$ containing $a$ is the number of $(\dim(A) - 1)$-dimensional subspaces of the $(\dim(V) - 1)$-dimensional space formed by removing $a$ from $V$, which is given by the $q$-binomial coefficient $\binom{\dim(V)-1}{\dim(A_R) - 1}_2$. The total number of possible $A_R$ is $\binom{\dim(V)}{\dim(A')}_2$, so 
        \[
            \Pr_{A_R}[a\in A_R] 
            = \frac{\binom{\dim(V)-1}{\dim(A_R) - 1}_2}{\binom{\dim(V)}{\dim(A_R)}_2}
            = \frac{2^{\dim(A)}-1}{2^{\dim(V)}-1}
            \leq 2^{\dim(A) - \dim(V) +1}
            = 2^{-3\secpar/6 + 1}
        \]
        Taking a union bound over the $2^{\dimTminusR}$ vectors in $A$, the probability that $T \intersect T' \neq R$ is at most $2^{-\secpar/6+1}$.
        \ifsubmission \qed \fi
    \end{proof}

    \begin{claim}\label{subclaim:poni-sec-vp-in-tp-s-x}
        There exists $q = \poly$ such that
        \[
            \Pr\left[\big((v_P-x_{\OSP}) \in (T_R + S) + x \big) \land \big(v_2 \in S^\perp + z\big)\right] \geq 1/q
        \]
    \end{claim}
    \begin{proof}
        Consider the following hybrid experiments.
        \begin{itemize}
            \item $\Hyb_0$ is the original $\PoNIMoE_{D_1,D_2}(\adv_1, \adv_P, \adv_H)$ game.

            Let $v_P$ be the vector sent by $\adv_P$ in the $\PoNI\langle \adv_P(\regP), V(T, x_T)\rangle$ execution and let $x_{\OSP}$ be the $X$ error from the underlying $\OSP$ protocol. The final $\PoNI$ is accepting if $(v_P - x_{\OSP}) \in T + x_T$.

            \item $\Hyb_1$ additionally samples a random $W\supset T$ with dimension $\dimW$ and changes $\adv_P$'s success condition to $(v_P - x_\OSP) \in W + x$. 
            
            Note that in this hybrid, $x_T$ is not used during the $\PoNI$ execution because the $\PoNI$ decision is replaced with a different success condition for $\adv_P$. Thus, the (decision-less) $\PoNI$ execution can also be written as $\PoNI\langle \adv_P(\regP), V(T, 0)\rangle$.

            \item $\Hyb_2$ replaces $T$ in the $\PoNI$ execution. Instead, the challenger samples $W\supset T$ as in $\Hyb_1$, then samples $T_W \gets \Subspc(W, \dimT)$ and runs $\PoNI\langle \adv_P(\regP), V(T_W, 0)\rangle$.

            \item $\Hyb_3$ replaces $(W, T)$ by $(T_R+S, T_R)$. This modifies both the $\PoNI$ execution and the success condition for $\adv_P$.
        \end{itemize}

        By assumption, the adversaries in $\Hyb_0$ simultaneously succeed with $1/p$ probability for some $p = \poly$. Since $W + x \supset T + x$, the winning probability can only increase in $\Hyb_1$:
        \[
            \Pr\left[\big((v_P-x_{\OSP}) \in W + x \big) \land \big(v_2 \in S^\perp + z\big)\right] \geq 1/p
        \]

        We use this fact to show that $\Hyb_2$ has simultaneous success probability $1/q_2$, for some $q_2 = \poly$, in \Cref{subsubclaim:poni-sec-coset-subclaim-h2}. Then, we use that fact to show that $\Hyb_3$ has simultaneous success probability $1/q_3$, for some $q_3 = \poly$, in \Cref{subsubclaim:poni-sec-coset-subclaim-h3}. The claim follows from the fact that $\Hyb_3$ is precisely the same distribution as our certified deletion adversary $(\adv'_1, \adv'_2)$.

        \begin{claim}\label{subsubclaim:poni-sec-coset-subclaim-h2}
            There exists $q_2 = \poly$ such that 
            \[
                \Pr_{\Hyb_2}\left[\big((v_P-x_{\OSP}) \in (T_R + S) + x \big) \land \big(v_2 \in S^\perp + z\big)\right] \geq 1/q_2
            \]
        \end{claim}
        \begin{proof}
            We invoke the simultaneous search lemma (\Cref{lem:simult-search}) using the following distribution and adversaries.
            $\calD$ samples $(\aux_1, \aux_H, S+x, T+x_T, R^\perp + z_R^\perp, z)$ as in the $\PoNIMoE_{D_1,D_2}$ game. Then, it samples $W\gets \Supspc(T, \dimW)$ and sets
            \begin{align*}
                \aux_A &\coloneqq \aux_1
                \\
                \aux_B &\coloneqq (\aux_H, S^\perp + z)
                \\
                \aux_C &\coloneqq T
                \\
                \cAuxClass &\coloneqq (W, x_W) \quad\text{ where } x_W = \Can_W(x)
            \end{align*}
            The ``fake'' distribution $\calD'(W, x_W)$ samples and outputs $T_W \gets \Subspc(W, \dimT)$.
            
            The search adversaries $(A, B, C)$ are as follows:
            \begin{itemize}
                \item $A(\aux_1)$ outputs $(\calH, \calP) \gets \adv_1(\aux_1)$.
                \item $B(\calH, \aux_H, S^\perp + z)$ runs $v_H \gets \adv_H(\calH, \aux_H)$ and outputs $1$ if $v_H \in S^\perp + z$.
                \item $C(\calP, T, (W, x_W))$ runs $\PoNI\langle \adv_P(\calP), V(T)\rangle$ until the decision step, obtaining $v_P$ and $x_\OSP$, then outputs $1$ if $(v_P - x_\OSP) \in W + x_W$.
            \end{itemize}

            Observe that $B$ and $C$ (when given the original distribution) simultaneously output $1$ whenever $\big((v_P-x_{\OSP}) \in W + x \big) \land \big(v_2 \in S^\perp + z\big)$ in $\Hyb_1$. Thus, they simultaneously output $1$ with probability $1/p$.

            The last precondition for the simultaneous search lemma is to show that 
            \[
                \{(\aux_1, T, (W, x_W))\} \approx_c \{(\aux_1, T_W, (W, x_W))\}
            \]
            We show this using the following series of hybrids experiments:
            \begin{itemize}
                \item $\Hyb_{2,0}$ is the distribution $(\aux_1, T, (W, x_W))$, where $(\aux_1, S, x, T)$ are distributed as in $\PoNI_{D_1,D_2}$ and $W\gets \Supspc(T, \dimW)$.
                \item $\Hyb_{2,1}$ replaces $(\aux_1, S, x, T)$ by uniformly sampling $(S,x,z, T)$ subject to $S\subset T$ and sampling $\aux_1 \gets \Sim(\ket{S_{x,z}})$. Since $S$, $T$, and $W$ are both uniform subject to $S\subset T \subset W$, it is equivalent to sample $W$ uniformly first, then $T\subset W$ uniformly, then $S\subset T$ uniformly.
                \item $\Hyb_{2,2}$ replaces $T$ by $T_W \gets \Subspc(W, \dimT)$.
                \item $\Hyb_{2,3}$ replaces $(\aux_1, S, x, T)$ by the original distribution from $\PoNIMoE_{D_1, D_2}$. Note that $T$ is still used to sample $W$.
            \end{itemize}
            $\Hyb_{2,1}$ is computationally indistinguishable from $\Hyb_{2,0}$ by the precondition on $(D_1, D_2)$. $\Hyb_{2,2}$ is statistically close to $\Hyb_{2,1}$ by \Cref{claim:PoNI-coset-extractor-indist-with-coset-state} (which operates for fixed $W$ and random $T\subset W$). $\Hyb_{2,3}$ is computationally indistinguishable from $\Hyb_{2,2}$ by the precondition on $(D_1, D_2)$.

            Therefore \Cref{lem:simult-search} implies the claim.
            \ifsubmission \qed \fi
        \end{proof}

        \begin{claim}\label{subsubclaim:poni-sec-coset-subclaim-h3}
            There exists $q_3 = \poly$ such that 
            \[
                \Pr_{\Hyb_3}\left[\big((v_P-x_{\OSP}) \in (T_R + S) + x \big) \land \big(v_2 \in S^\perp + z\big)\right] \geq 1/q_2
            \]
        \end{claim}
        \begin{proof}
            We invoke the simultaneous search lemma (\Cref{lem:simult-search}) using the following distribution and adversaries.
            $\calD$ samples $(\aux_1, \aux_H, S+x, T+x_T, R^\perp + z_R^\perp, z)$ as in the $\PoNIMoE_{D_1,D_2}$ game. Then, it samples $W\gets \Supspc(T, \dimW)$ and $T_W \gets \Supspc(W, \dimT)$ and sets
            \begin{align*}
                \aux_A &\coloneqq \aux_1
                \\
                \aux_B &\coloneqq (\aux_H, S^\perp + z)
                \\
                \aux_C &\coloneqq (W, T_W)
                \\
                \cAuxClass &\coloneqq (S, x)
            \end{align*}
            The ``fake'' distribution $\calD'(S,x)$ samples $R\gets \Subspc(S, \dimR)$ and $T_R \gets \Supspc(R, \dimT)$, then outputs $(T_R + S, T_R)$.
            
            The search adversaries $(A, B, C)$ are as follows:
            \begin{itemize}
                \item $A(\aux_1)$ outputs $(\calH, \calP) \gets \adv_1(\aux_1)$.
                \item $B(\calH, \aux_H, S^\perp + z)$ runs $v_H \gets \adv_H(\calH, \aux_H)$ and outputs $1$ if $v_H \in S^\perp + z$.
                \item $C(\calP, (W, T_W), (S,x))$ runs $\PoNI\langle \adv_P(\calP), V(T_W)\rangle$ until the decision step, obtaining $v_P$ and $x_\OSP$, then outputs $1$ if $(v_P - x_\OSP) \in W + x$.
            \end{itemize}

            Observe that $B$ and $C$ (when given the original distribution) simultaneously output $1$ whenever $\big((v_P-x_{\OSP}) \in W + x \big) \land \big(v_2 \in S^\perp + z\big)$ in $\Hyb_2$. Thus, they simultaneously output $1$ with probability $1/q_2$ by \Cref{subsubclaim:poni-sec-coset-subclaim-h2}.

            The last precondition for the simultaneous search lemma is to show that 
            \[
                \{(\aux_1, (W, T_W), (S, x))\} \approx_c \{(\aux_1, (T_R + S, T_R), (S, x))\}
            \]
            We show this using the following series of hybrids experiments:
            \begin{itemize}
                \item $\Hyb_{2,0}$ is the distribution $(\aux_1, T, (W, x_W))$, where $(\aux_1, S, x, T)$ are distributed as in $\PoNI_{D_1,D_2}$ and $W\gets \Supspc(T, \dimW)$.
                \item $\Hyb_{2,1}$ replaces $(\aux_1, S, x, T)$ by uniformly sampling $(S,x,z, T)$ subject to $S\subset T$ and sampling $\aux_1 \gets \Sim(\ket{S_{x,z}})$. Since $S$ and $W$ are both uniform subject to $S\subset T \subset W$, it is equivalent to sample $S$ uniformly first, then $W\supset S$ uniformly. Note that $T$ is not used in this hybrid, so there is no need to explicitly sample it.
                
                \item $\Hyb_{2,2}$ replaces $(W, T_W)$ by $(T_R + S, T_R)$ where $R\gets \Subspc(S, \dimR)$ and $T_R \gets \Supspc(R, \dimT)$.
                
                \item $\Hyb_{2,3}$ replaces $(\aux_1, S, x, T)$ by the original distribution from $\PoNIMoE_{D_1, D_2}$. Note that $T$ is still used to sample $W$.
            \end{itemize}
            $\Hyb_{2,1}$ is computationally indistinguishable from $\Hyb_{2,0}$ by the precondition on $(D_1, D_2)$. $\Hyb_{2,2}$ is statistically close to $\Hyb_{2,1}$ by \Cref{claim:poni-coset-extractor-1sample} (which operates for fixed $S$ and random $W\supset S$). $\Hyb_{2,3}$ is computationally indistinguishable from $\Hyb_{2,2}$ by the precondition on $(D_1, D_2)$.

            Therefore \Cref{lem:simult-search} implies the claim.
            \ifsubmission \qed \fi
        \end{proof}
        \ifsubmission \qed \fi
    \end{proof}
    \ifsubmission \qed \fi
    \end{proof}
\fi
\section{Encryption with Proofs of No-Intrusion}\label{sec:poni-enc}

\subsection{Definition}\label{sec:poni-enc:defs}

\begin{definition}[Encryption with Proofs of No-Intrusion]
    A (public-key) \textbf{encryption scheme with a proof of no-intrusion} consists of the following QPT algorithms.
    \begin{itemize}
        \item $(\pk, \sk) \gets \KeyGen(1^\secpar)$ takes as input the security parameter $1^\secpar$ then outputs a public key $\pk$ and a secret key $\sk$.
        \item $\verkey \gets \VerKeyGen(1^\secpar)$ takes as input the security parameter $1^\secpar$ then outputs a verification key $\verkey$.
        \item $\ct \gets \Enc(\pk, m, \verkey)$ takes as input a public key $\pk$, a message $m$, and a verification key $\verkey$, then outputs a ciphertext $\ct$.
        \item $m\gets \Dec(\sk, \ct)$ takes as input a secret key $\sk$ and a ciphertext $\ct$, then outputs a message $m$.
    \end{itemize}

    Additionally, it is equipped with a \textbf{proof of no-intrusion}, which is an interactive protocol \emph{with classical communication} between a QPT prover $P$ holding a state $\ket{\psi}$ and a PPT verifier holding a verification key $\verkey$:
    \[
        (\ket{\psi'}, b) 
        \gets
        \PoNI\langle \prover(\ket{\psi}), \verifier(\verkey)\rangle
    \]
    At the end of the protocol, $P$ outputs a state $\ket{\psi'}$ and $V$ outputs a decision bit $b\in \{\Accept, \Reject\}$.

    The scheme must satisfy the following properties:
    \begin{itemize}
        \item \textbf{Correctness.} For every message $m$, every key pair $(\pk, \sk)$ in the support of $\KeyGen(1^\secpar)$, and every verification key $\verkey$ in the support of $\VerKeyGen(1^\secpar)$,
        \[
            \Pr\left[m = \Dec(\sk, \Enc(\pk, m, \verkey)\right] = 1-\negl
        \]

        \item \textbf{Semantic Security.} For every pair of messages $(m_0, m_1)$,
        \[
            \left\{
                (\ct_0, \verkey) : 
                \begin{array}{c}
                    (\pk, \sk) \gets \KeyGen(1^\secpar) \\
                    \verkey\gets \VerKeyGen(1^\secpar) \\
                     \ct_0 \gets \Enc(\pk, m_0, \verkey)  
                \end{array} 
            \right\}
            \approx_c
            \left\{
                (\ct_1, \verkey) :
                \begin{array}{c}
                (\pk, \sk) \gets \KeyGen(1^\secpar) \\
                \verkey\gets \VerKeyGen(1^\secpar) \\
                 \ct_1 \gets \Enc(\pk, m_1, \verkey)  
                \end{array} 
            \right\}
        \]
    
        \item \textbf{PoNI Correctness and State Preservation:} Let $(\pk, \sk)$ be in the support of $\KeyGen(1^\secpar)$, let $m$ be an arbitrary message, let $\verkey$ be in the support of $\VerKeyGen$, and let $\ket{\ct}$ be a ciphertext which decrypts to $m$ with certainty.
        \[
            (\ket{\psi'}, b) \gets \PoNI\langle P(\ket{\ct}),\ V(\verkey)\rangle 
        \]
        $b = \Accept$ and $\ket{\psi}$ belongs to the space of ciphertexts which decrypt to $m$.
        
        \item \textbf{PoNI Security:} The construction satisfies either search PoNI security (\Cref{def:PoNI-security-search}) or decision PoNI security (\Cref{def:PoNI-security-decision}).
    \end{itemize}    
\end{definition}

A nuance in the definition of state preservation is that the ciphertext is allowed to \emph{change} as a result of the PoNI, as long as it still decrypts to the same message.

We emphasize that $\verkey$ can be generated once and used for many ciphertexts -- even under different public keys. This ensures that any client sending ciphertexts does not have to store an amount of data that grows with the number of ciphertexts sent. After all, if the client had to store the same amount of data locally, why should the client not store the original data itself?

\justin{Commentary on security notions.}
For security, we consider a general scenario where 

\begin{definition}[Proofs of No-Intrusion for Encryption: Search Security]\label{def:PoNI-security-search}
    Consider the following security game $\PoNIEncSearch_{n}(\adv)$, played by an adversary $\adv = (\adv_1, \adv_P, \adv_H)$ consisting of three QPT algorithms (with auxiliary quantum input) and parameterized by a non-negative integer $n$.
    \begin{enumerate}
        \item Sample a message $m\gets \{0,1\}^\secpar$, a key pair $(\pk, \sk)\gets \KeyGen(1^\secpar)$, a verification key $\verkey \gets \VerKeyGen(1^\secpar)$, and a ciphertext $\ct \gets \Enc(\pk, m, \verkey)$.
        \item Initialize $\adv_1$ with $(\pk, \ct)$.
        \item Perform $\PoNI\langle \adv_1, \verifier(\verkey)\rangle$ a total of $n$ times. If the verifier outputs $\Reject$ in any of these executions, the adversary immediately loses (output $0$).
        \item $\adv_1$ outputs two registers $(\calH, \calP)$.
        \item Perform $b \gets \PoNI\langle \adv_P(\calH), \verifier(\verkey)\rangle$, where $b$ is the verifier's decision bit. 
        \item Run $m' \gets \adv_H(\regH, \sk, \vk)$.
        \item Output $1$ (the adversary wins) if $b = \Accept$ and $m' = m$. Otherwise output $0$ (the adversary loses).
    \end{enumerate}

    \noindent We say the PoNI has $n$-time search security if for all QPT adversaries $\adv$,
    \[
        \Pr[1\gets \PoNIEncSearch_n(\adv)] = \negl
    \]
    If this holds for all $n = \poly$, then we simply say that the PoNI has search security.

    \noindent If this holds for QPT $\adv_1$ and $\adv_P$, but \emph{unbounded} $\adv_H$, we say that the PoNI has ($n$-time) \emph{everlasting} search security.
\end{definition}

\justin{Commentary on difference between search and decision.}
\ifsubmission
We define a decisional version and explain the differences from the search version in \Cref{app:decisional-encryption}.
\else
\begin{definition}[Proofs of No-Intrusion for Encryption: Decisional Security]\label{def:PoNI-security-decision}
    Consider the following security game $\PoNISec_{n}(m)$, played by an adversary $\adv(\ket{\psi})$ and parameterized by a message bit $m\in \{0,1\}$ along with a non-negative integer $n$.
    \begin{enumerate}
        \item Sample a key pair $(\pk, \sk)\gets \KeyGen(1^\secpar)$, a verification key $\verkey\gets \VerKeyGen(1^\secpar)$ and a ciphertext $\ct \gets \Enc(\pk, m, \verkey)$.
        \item Initialize $\adv(\ket{\psi})$ with $(\pk, \ct)$.
        \item Perform $\PoNI\langle \adv(\cdots), \verifier(\verkey)\rangle$ a total of $n$ times. If the verifier outputs $\Reject$ in any of these executions, immediately output $\bot$.
        \item $\adv$ outputs a register $\regR_{\Dec}$.
        \item Perform $\PoNI\langle \adv(\cdots), \verifier(\verkey)\rangle$. If the prover rejects, output $\bot$. Otherwise output $(\regR_{\Dec}, \sk, \verkey)$.
    \end{enumerate}

    \noindent We say the PoNI has $n$-time security if for all QPT adversaries $\adv(\ket{\psi})$,
    \[
        \{\PoNISec_{n}(0)\} \approx \{\PoNISec_{n}(1)\}
    \]
    If this holds for all $n = \poly$, then we simply say that the PoNI is secure.

    \noindent If these two distributions are instead \emph{statistically} close, we say that the PoNI is ($n$-time) everlasting secure.
\end{definition}
\fi

\justin{Nuances with the definition: Mixing a search-based security notion (PoNI acceptance) with a decisional one. So can't just use MoE-style ``simultaneously correct'' definition. Could be catastrophic if search-security was violated with probability 1/2 and decisional was violated with probability 1 - but this would be absolutely fine for 2 decisional. Could try to do something like "conditioned on search passing, decisional is 1/2". However, then would need to also include an ``or'' requirement that if this is not the case, search violation is negligible. This style of definition wraps both up together nicely.}

Proofs of non-intrusion are related to another property called certified deletion. Certified deletion allows a sender to encrypt a messages in a way that allows the ciphertext to be destructively measured to produce a certificate. If the prover accepts the certificate, they are guaranteed that whoever holds the residual state of the ciphertext cannot learn any information about the message even if the secret key is stolen (or potentially the hardness assumption is broken). It is not hard to see that proofs of no-intrusion imply certified deletion with quantum certificates -- simply return the whole ciphertext and the verifier can check it by running the proof of non-intrusion internally.

\begin{claim}
    Any encryption scheme with a a decisional proof of no-intrusion is also an encryption scheme with certified deletion where the certificate is quantum.
\end{claim}

On the other hand, certified deletion does \emph{not} imply a proof of no-intrusion. The key difference is that proofs of no-intrusion require state preservation -- that is, any tested ciphertext should still decrypt to the same message. If state preservation were not required, then the prover could simply verifiably delete their ciphertext, rendering it useless forevermore.


\subsection{Construction}

Our construction is based on \cite{EC:BGKMRR24}'s construction of certified deletion using coset states. 
The main modification is how we mask the message. The prior work's construction encrypted $\langle z, \vec{1}\rangle \oplus m$, where $m$ is a single bit message. Our construction will instead mask a multi-bit message $m$ by the full information of $z$ (albeit in a compressed form). Specifically, we encrypt $(M_S z) \oplus m$, where $M_S$ is a matrix whose rows form a basis for $S$. Looking forward, we can recover $v\in S^\perp + z$ from any adversary who finds $M_S z$ because the kernel of $M_S$ is precisely $S^\perp$.

To extend the construction support proofs of non-intrusion, we use the proof of no-intrusion for coset states developed in \cref{sec:poni-coset}. Since that construction incurs a random Pauli $Z$ error which is learned by the verifier, we also include an FHE ciphertext containing a correction factor when encrypting a message $m$. Over subsequent PoNI tests, the verifier can help the holder of the ciphertext homomorphically update the correction factor.

\begin{construction}[Encryption with Proofs of No-Intrusion]\label{constr:enc-poni}
    The construction uses a public-key encryption scheme $\PKE$, a fully-homomorphic encryption scheme $\FHE$, and the proof of no intrusion $\PoNI_{\coset}$ from \Cref{constr:coset-poni}. The core algorithms are as follows.
    \begin{itemize}
        \item $\KeyGen(1^\secpar)$: Generate $(\sk, \pk) \gets \PKE.\KeyGen(1^\secpar)$. Generate $(\sk, \pk) \gets \PKE.\KeyGen(1^\secpar)$. Define $\sk \coloneqq (\sk_{\PKE}, \sk_{\FHE}$ and $\pk \coloneqq (\pk_{\PKE}, \pk_{FHE})$.\footnote{If one is willing to use $\FHE$ in place of the $\PKE$ scheme, a single key can be used.} Output $(\sk, \pk)$.

        \item $\VerKeyGen(1^\secpar)$: Sample a PRF key $k_{\PRF}\gets \{0,1\}^\secpar$ and initialize a $\lambda$-bit counter $\ctr = 1$. Output $(k_{\PRF}, \ctr)$.
        
        \item $\Enc(\pk, m\in \{0,1\}^{\dimS}; \verkey)$: 
        \begin{enumerate}

            \item Parse $\verkey = (k_\PRF, \ctr)$. Sample a $\dimS$-dimensional subspace $S$ and a vector $x \in \co(S)$ using randomness $\PRF(k_{\PRF}, \ctr)$. We abuse notation in the future by writing $(S, x) = \PRF(k_{\PRF}, \ctr)$. Update $\ctr$ to $\ctr+1$.\footnote{Technically, one could use any seed so long as they do not use the same seed twice or let it depend on $k_{\PRF}$. We suggest using a counter because it makes it easy to track the number of messages which a sender might need to test in the event of a data hack (and nobody has time to encrypt $2^\secpar$ messages).}
            \item Sample a vector $z\gets \co(S^\perp)$ using \emph{fresh} randomness.
            
            \item Let $M_S \in \bbF_2^{\dimS \times \secpar}$ be a matrix whose rows span $S$.
            \item Output
            \[
                \ct = \left(\begin{array}{c}
                      \ket{S_{x,z}},\\
                      \PKE.\Enc(\pk_{\PKE}, (M_S, M_S z \oplus m)),\\
                      \FHE.\Enc(\pk_{\FHE}, 0)
                \end{array}\right)
            \]
        \end{enumerate}
        

        \item $\Dec(\sk, \ct)$: 
        \begin{enumerate}
            \item Parse $\ct = (\ket{\psi},\ \ct_1,\ \ct_2)$.
            \item Decrypt $(M_S, v)\gets \PKE.\Dec(\sk_{\PKE}, \ct_1)$ and the correction factor $z_{\mathsf{cor}} \gets \FHE.\Dec(\sk_{\FHE}, \ct_2)$.
            \item Measure $\ket{\psi}$ in the Hadamard basis to obtain a vector $z'$.
            \item Output $M_S (z' - z_{\mathsf{cor}}) \oplus v$.
        \end{enumerate}
    \end{itemize}
    The \textbf{proof of no-intrusion} works as follows. Suppose the verifier wants to verify the $i$'th ciphertext they encrypted.
    \begin{enumerate}
        \item \textbf{Verifier:} Parse $\verkey = (k_\PRF, \ctr)$. Compute $(S, x) \gets \PRF(k_\PRF, i)$. Sample $T + x_T \supset S + x$.
        
        \item \textbf{Prover:} Parse $\ct = (\ket{\psi},\ \ct_1,\ \ct_2)$.
        
        \item Execute
        \[
            (\ket{\psi'}, (b, (0,z'))) \gets \PoNI_{\mathsf{coset}}\langle \prover(\ket{\psi}), \verifier(T, x_T)\rangle
        \]
        
        \item \textbf{Verifier:} Send $\FHE.\Enc(\pk_\FHE, z')$ to the prover. Output $b$.
        
        \item \textbf{Prover:} Homomorphically evaluate the function $f_{z'}(a) = a+z'$ (over $\bbF_2^{\secpar}$) on the correction ciphertext $\ct_2$ to obtain $\ct_2'$. Update $\ct$ to $(\ket{\psi'},\ \ct_1,\ \ct'_2)$.
    \end{enumerate}
\end{construction}

\begin{theorem}
    Assuming oblivious state preparation for coset states and fully homomorphic encryption, there exists an encryption scheme with search proofs of no intrusion.
\end{theorem}
\begin{proof}
    We first show the (mostly) standard correctness property. The only change is the syntax. A ciphertext takes the form
    \[
        \ct = \left(\begin{array}{c}
                      \ket{S_{x,z+z_\mathsf{cor}}},\\
                      \PKE.\Enc(\pk_{\PKE}, (M_S, M_S z \oplus m)),\\
                      \FHE.\Enc(\pk_{\FHE}, z_\mathsf{cor})
                \end{array}\right)
    \]
    By correctness of the underlying $\PKE$ and $\FHE$, the decryptor correctly obtains $M_S$, $z_\mathsf{cor}$, and $v = M_S z \oplus m$. Measuring $\ket{S_{x,z+z_\mathsf{cor}}}$ in the Hadamard basis results in a vector $z'\in S^\perp + z + z_\mathsf{cor}$. Then $v \oplus M_S(z' - z_\mathsf{cor}) = M_S z \oplus m \oplus M_S (z + s)$ for some $s\in S$. Since $S$ is the kernel of $M_S$ and the matrix multiplication is over $\bbF_2$, the output is $m$.

    Semantic security, even given $\verkey$, follows from the semantic security of the underlying $\PKE$. Otherwise, we could distinguish between $\PKE$ encryptions of $M_S z \oplus m_0$ and $M_S z \oplus m_1$, for fixed $M_S$, $z$, $m_0$, and $m_1$.

    Correctness and state preservation of the PoNI follow directly from the properties of \cref{constr:coset-poni} given in \cref{thm:poni-coset} as well as the correctness of FHE evaluation.

    Finally, we show search PoNI security for encryption
   by reducing to the properties of \cref{constr:coset-poni} given in \cref{thm:poni-coset}. Suppose that some adversary $(\adv_1, \adv_P, \adv_H)$ violated the search PoNI security of our encryption scheme. 
    Consider a hybrid experiment where the adversary does not lose if one of the initial $n$ PoNIs rejects.  $(\adv_1, \adv_P, \adv_H)$ wins this hybrid experiment with at least the probability that they win $\PoNIEncSearch$.

    $\adv_1$'s auxiliary information $\aux_1$ consists of the leftover information after receiving $\pk$ and $(\ket{S_{,x,z}}, \Enc(M_S, M_S z \oplus m))$, then running as the prover in $n$ PoNIs. Since $m$ is uniformly random, it statistically masks $M_S z$, i.e. the ciphertext can be written as $\Enc(S, r)$, independently of $z$.
    Turning to the computational requirements on the auxiliary information, $\Enc(S, r)$ is computationally indistinguishable from $\Enc(0)$.
    Furthermore, in this hybrid experiment, it is not necessary to decode the results of the initial $n$ PoNIs. So, by the semantic security property from \Cref{thm:poni-coset}, the initial $n$ PoNIs can be indistinguishably replaced by the verifier using an arbitrary $T'$ instead of $T\supset S$. Finally, $S$ is uniformly random within $T$, which is generated by a PRF. $S$ and $x$ are computationally indistinguishable from random by PRF security. Therefore $(\aux_1, S, x, T, R)$ is indistinguishable from sampling $S, x, T, R$ uniformly at random, encrypting $0$, and running the initial $n$ PoNIs using an arbitrary $T'$, so the preconditions of \Cref{thm:poni-coset} are satisfied.

    Define $\adv'_H$ as follows. It runs $\adv_H(\regH, (\sk, \vk))$ to obtain $m'$. Then, it decrypts $\Enc(M_S, r)$ to obtain $M_S$ and $r$. Finally, it computes a vector $v_H$ such that $M_S v_H = (r \oplus m')$. If $m'=m$, then $(r \oplus m) = M_S z$, so $M_S v_H = M_S z$. Since the kernel of $M_S$ is precisely $S^\perp$, this implies that $v_H \in S^\perp + z$.

    Thus, if $(\adv_1, \adv_P, \adv_H)$ violated the PoNI search security of our encryption scheme, then with noticeable probability $\adv'_H$ outputs $v_H \in S^\perp + z$ simultaneously with $\adv_P$ passing the $\PoNI$ execution. This contradicts the PoNI search security from \Cref{thm:poni-coset}.
    \ifsubmission \qed \fi
\end{proof}

\section{Proofs of No-Intrusion from Indistinguishability Obfuscation}
\justin{Read over this.}

In this section, we describe several other primitives which can be equipped with proofs of no-intrusion. This section will focus on primitives for which we know how to construct proofs of no-intrusion using indistinguishability obfuscation. We leave the question of reducing the assumptions to future work.

\subsection{Encryption with Proofs of No-Intrusion for Secret Keys}

An encryption scheme with proofs of no-intrusion for secret keys has similar syntax to single-decryptor encryption \cite{C:CLLZ21}. Additionally, it is equipped with a proof of no-intrusion.

\begin{definition}
    An \textbf{encryption scheme with proofs of no-intrusion for secret keys} is a tuple of QPT algorithms $(\KeyGen, \mathsf{QKeyGen}, \Enc, \Dec)$ together with an interactive protocol $\PoNI$ between a quantum prover and a classical verifier. The QPT algorithms act as follows.
    \begin{itemize}
        \item $(\pk, \msk) \gets \KeyGen(1^\secpar)$ takes as input the security parameter $1^\secpar$ then outputs a public key $\pk$ and a master secret key $\msk$.
        \item $\rho_\sk \gets \mathsf{QKeyGen}(\msk)$ takes as input a master secret key $\msk$ then outputs a quantum decryption key $\rho_{\sk}$.
        \item $\ct \gets \Enc(\pk, m)$ takes as input a public key $\pk$ and a message $m$, then outputs a classical ciphertext $\ct$.
        \item $m' \gets \Dec(\rho_\sk, \ct)$ takes as input a quantum decryption key $\rho_{\sk}$ and a ciphertext $\ct$, then outputs a classical message $m$.
    \end{itemize}
    The encryption scheme must satisfy correctness: for every message $m$,
    \[
        \Pr\left[\Dec(\rho_\sk, \ct) = m \middle| \begin{array}{c}
              (\pk, \sk) \gets \KeyGen(1^\secpar) \\
             \rho_\sk \gets \mathsf{QKeyGen}(\msk) \\
             \ct \gets \Enc(\pk, m)
        \end{array}\right] = 1 -\negl
    \]
    Furthermore the proof of no intrusion must satisfy the following properties:
    \begin{itemize}
        \item \textbf{PoNI Correctness and State Preservation.} For every message $m$,
        \[
            \Pr\left[
            \begin{array}{c}
                 \Dec(\rho'_\sk, \ct) = m  \\
                 \land\ b = \Accept
            \end{array}
            \middle| \begin{array}{c}
              (\pk, \sk) \gets \KeyGen(1^\secpar) \\
             \ct \gets \Enc(\pk, m) \\
             \rho_\sk \gets \mathsf{QKeyGen}(\msk) \\
             (b, \rho'_\sk) \gets \PoNI\langle P(\rho_\sk), V(\msk)\rangle
        \end{array}\right] = 1 -\negl
        \]

        \item \textbf{PoNI Security.} See \Cref{def:poni-sec-decryption}.
    \end{itemize}
    If the verifier only requires $\pk$ to run $\PoNI$ (instead of $\msk$), we say it has a public $\PoNI$.
\end{definition}

To define PoNI security, we first need the notion of a $\gamma$-good decryptor from \cite{C:CLLZ21}.
\justin{Use this definition or the one from more recent work?}
\begin{definition}[$\gamma$-Good Decryptor]
    Let $\gamma \in [0,1]$. Let $m_0$ and $m_1$ be a pair of messages. Let $(\pk, \sk)$ be a key pair for an encryption scheme. Let $(U, \rho)$ be a quantum adversary and define the following POVM $(\calP, \calQ)$ which is a mixture of projective measurements:
    \begin{enumerate}
        \item Sample $b\gets \{0,1\}$. Encrypt $\ct_b \gets \Enc(\pk, m_b)$.
        \item Run $(U, \rho)$ on input $\ct_b$. Output $1$ if the adversary guesses $b$ correctly, and $0$ otherwise.
    \end{enumerate}
    We say a quantum adversary (potentially with auxiliary quantum state) is a $\gamma$-good with respect to $(\pk, m_0, m_1)$ if for $\delta = \omega(\log(\secpar))$ and every $\epsilon = 1/\poly$,
    \[
        \Pr[1\gets \ATI_{\calP, 1/2 + \gamma}^{\epsilon, \delta}(\rho)] = 1-\negl
    \]
\end{definition}

Security intuitively says that conditioned on a prover $P$ given an accepting PoNI, no hacker could have stolen a decryption key -- any new ciphertexts have full semantic security.

\begin{definition}[PoNI Security for Decryption Keys]\label{def:poni-sec-decryption}
    Consider the following security game $\PoNISec_{n, \gamma}(\adv)$, played by 3-part adversaries $\adv = (\adv_1, \adv_P, \adv_D)$ and parameterized by a non-negative integer $n$ and $\gamma \in [0,1/2]$.
        \begin{enumerate}
            \item Sample $(\pk, \sk) \gets \KeyGen(1^\secpar)$ and $\rho_\sk \gets \mathsf{QKeyGen}(\msk)$.
            \item Initialize $\adv_1$ with $(\pk, \rho_\sk)$.
            \item Execute $\PoNI\langle \adv_1, V(\msk)\rangle$ a total of $n-1$ times. If the verifier rejects in any of them, the adversary immediately loses (output $0$).
             $\adv_1$ outputs two registers $\regP$ and $\regD$.
            \item Run $\PoNI\langle \adv_P(\regP), V(\msk) \rangle$. Let $d_P$ be the verifier's decision bit.
            
            \item $\adv_D$ is initialized with $\regD$. 

            \item The adversary wins (output $1$) if $d_P = \Accept$ 
            and simultaneously $\adv_D(\regD)$ is a $\gamma$-good decryptor.
        \end{enumerate}
        We say the PoNI has $n$-time $\gamma$-security if for every QPT $(\adv_1, \adv_P, \adv_D)$,
        If this holds for all $n = \poly$ and all $\gamma = 1/\poly$, then we simply say that the PoNI is secure.
\end{definition}

We observe that \emph{any} public-key single-decryptor encryption scheme has public proofs of no intrusion. The proof of no-intrusion is to encrypt a random message, then ask the prover to decrypt it.

 As a corollary, encryption with proofs of no-intrusion for secret keys exist assuming post-quantum indistinguishability obfuscation and one-way functions~\cite{eprint:KY25}.

 \subsection{Signature Tokens with Proofs of No-Intrusion}

\cite{Q:BS23} proposed the idea of quantum signature token and constructed it in an oracle model. A quantum signature token allows Alice to provide Bob with a way to sign one, \emph{and only one} message under her public key. In particular, Bob can choose the message \emph{after} interacting with Alice. 

\cite{C:CLLZ21} showed how to construct signature tokens using just indistinguishability obfuscation and one-way functions. At a high level, a signature token for a one-bit message is a single coset state $\ket{S_{x,z}}$. The public key consists of obfuscated membership programs $C_{S+x}$ and $C_{S^\perp + z}$, which test membership of a vector in $S+x$ and $S^\perp + z$, respectively. To sign $0$, measure the coset state in the computational basis to get a vector $v_0 \in S+x$. This can be verified using the membership program for $S+x$. Similarly, a signature on $1$ is a vector in $S^\perp + z$, obtained by measuring the coset state in the Hadamard basis. It can be checked using the membership program for $S^\perp + z$.

\paragraph{Proofs of No-Intrusion for Signature Tokens.}
As an extension of signature tokens, it would be quite useful if Alice were able to test whether Bob has lost the signing token she gave him. Bob could prove that no one else can use his token to sign message on Alice's behalf in one of two ways. One option is for him to sign a message right there and then. However, this would prevent him from using his signature token to sign a message of his choice. The other is for him to give a proof of no-intrusion, proving that no one could have stolen his signing tokens.

The coset state-based construction is particularly amenable to proofs of no-intrusion. In \cite{C:CLLZ21}'s construction of single-decryptor encryption, the challenge to decrypt a random ciphertext is essentially a test that the prover/decryptor still has the original coset state. Thus, a ciphertext from their scheme could be used to test whether Bob still has the signing token from Alice.

\fi

\bibliographystyle{alpha}
\bibliography{bib/abbrev3,bib/crypto,bib/project-bib}

\appendix
\section{Proof of Certified Deletion for Coset States}\label{app:cd}

For convenience, we restate \Cref{thm:search-cd} here.

\begin{definition}{Certified Deletion Search Game.} The certified deletion game $\CDGame_{D_1, D_2}(\adv_1, \adv_2)$ is parameterized by two distributions $D_1$ and $D_2$ and is played by two adversaries $\adv_1$ and $\adv_2$. 
$D_1$ is a distribution over $(S + x, T, R^\perp + z_{R^\perp}, k)$ where $R\subset S \subset T \subset \bbF_2^\secpar$ are subspaces with dimensions $d_R$, $d_S$, and $d_T$, respectively, where $x\in \co(S)$, and where $z_{R^\perp} \in \co(R^\perp)$.
$D_2$ is a distribution over (potentially entangled) quantum states $(\aux_1, \aux_2)$. The game is played as follows.
\begin{enumerate}
    \item Sample $(S + x, T, R^\perp + z_{R^\perp}, k) \gets D_1$. Sample $z \gets \co(S^\perp) \intersect R^\perp + z_{R^\perp}$ and let $x_T = \Can_T(x)$.
    \item Sample $(\aux_1, \aux_2) \gets D_2(\ket{S_{x,z}}, S+x, T + x_T, R^\perp + z_{R^\perp}, k)$.
    \item Run $(v_1, \rho) \gets \adv_1(\aux)$.
    \item Run $v_2 \gets \adv_2(\rho, \aux_2)$.
    \item $\adv_1$ and $\adv_2$ win if $v_1 \in S+ x$ and $v_2\in S^\perp + z$.
\end{enumerate}
\end{definition}

\begin{theorem}
    If there exists a QPT algorithm $\Sim$ such that
    \begin{gather*}
        \left\{
            \aux_1 : \begin{array}{c}
                (S+x, T, R^\perp + z_{R^\perp}, k) \gets D_1
                \\
                 (\aux_1, \aux_2) \gets D_2\big(\ket{S_{x,z}}, S+x, T + x_T, R^\perp + z_{R^\perp}, k \big)
            \end{array}
        \right\}
        \\
        \approx_c
        \\
        \left\{
            \aux_1 : \begin{array}{c}
                R, S, T \gets \Subspc(\bbF_2^\secpar, (d_R, d_S, d_T)) \text{ s.t. } R\subset S \subset T
                \\
                x \gets \co(S),\ z\gets \co(S^\perp),\ x_T = \Can_T(x),\ z_{R^\perp} = \Can_{R^\perp}(S)
                \\
                 \aux_1 \gets \Sim\big(\ket{S_{x,z}}, T + x_T, R^\perp + z_{R^\perp}\big)
            \end{array}
        \right\}
    \end{gather*}
    and $D_1$ is supported on $(S, T, R)$ such that $2^{-(d_S - d_R)} = \negl$ and $2^{-(d_T - d_S)} = \negl$, then for all QPT $\adv_1$ and (unbounded) quantum $\adv_2$,
    \[
        \Pr[\mathsf{win}\gets \CDGame_{D_1, D_2}(\adv_1, \adv_2)] = \negl
    \]
\end{theorem}
\begin{proof}
    Consider the following hybrid experiments.
    \begin{itemize}
        \item $\Hyb_0$ is the original $\CDGame_{D_1,D_2}(\adv_1, \adv_2)$ game.
        \item $\Hyb_1$ changes how $z$ and $\ket{S_{x,z}}$ are generated. After sampling $(S + x, T, R^\perp + z_{R^\perp}, k) \gets D_1$, prepare the state
        \[
            \ket{\psi_0} \propto \sum_{z \in (\co(S^\perp)\intersect R^\perp) + z_{R^\perp}} \ket{S_{x,z}}_{\calS} \otimes \ket{z}_{\mathsf{prep}}
        \]
        Measure the $\mathsf{prep}$ register in the computational basis to obtain $z$, then generate $(\aux_1, \aux_2) \gets D_2(\calS,\dots)$, substituting the state in register $\calS$ for $\ket{S_{x,z}}$.

        \item $\Hyb_2$ delays the measurement of the $\mathsf{prep}$ register until after $v_1$ and $v_2$ are determined.

        \item $\Hyb_3$ inserts some computation before initializing $\adv_1$ and after $\adv_2$ outputs. After producing $\ket{\psi_0}$, apply the unitary $U$ swapping
        \[
            \ket{z} \leftrightarrow \sum_{x_R\in (\co(R)\intersect S) + x} (-1)^{z\cdot x_R} \ket{x_R}_{\mathsf{prep}}
        \]
        for every $z \in (\co(S^\perp)\intersect R^\perp) + z_{R^\perp}$, and acting as the identity elsewhere, to register $\mathsf{prep}$. \cite{EC:BGKMRR24} showed that this results in the state
        \[
            \ket{\psi_1} \propto \sum_{x_R\in (\co(R)\intersect S) + x} \ket{R_{x_R,z}}_{\calS} \otimes \ket{x_R}_{\mathsf{prep}}
        \]
        Then, after $\adv_2$ outputs $v_2$, apply $U^\dagger = U$ to register $\mathsf{prep}$ and measure $z$ from it.

        \item $\Hyb_4$ changes the winning condition and inserts a measurement in the outcome step after $v_1$ and $v_2$ are determined but before $U^\dagger$ is applied. 
        
        Measure the $\mathsf{prep}$ register with respect to the PVM $\{\Pi_{v_1}, I - \Pi_{v_1}\}$ where 
        \[
            \Pi_{v_1} = \sum_{\substack{x_R\in (\co(R)\intersect S) + v_1\\\text{s.t. }v_1 \in R + x_R}} \ketbra{x_R}
        \]
        Then, apply $U^\dagger$ to $\mathsf{prep}$, measure it to obtain $z$. The adversary wins if the measurement result is $\Pi_{v_1}$ and $v_2\in R^\perp + z_{R^\perp}$.
    \end{itemize}

    $\Pi_{v_1}$ projects onto a unique $\ket{x_R}$ because $R + x_R \neq R + x'_R$ for all distinct $x_R, x'_R  \in (\co(R)\intersect S) + x$ and $v_1 \in R + x_R$ implies $R+ x_R = R + v_1$.
    Conditioned on this check passing, $z$ is the result of measuring $U \ket{x_R}$ in the computational basis. This is a uniformly random element of $(\co(R)\intersect S) + x$, independently of $v_2$, so
    \[
        \Pr[\mathsf{win} \gets \Hyb_4] \leq 2^{-(d_S - d_R)} = \negl
    \]

    It remains to be shown that $\Pr[\mathsf{win}\gets \Hyb_0]$ is negligibly far from this.
    \begin{align*}
        \Pr[\mathsf{win}\gets \Hyb_0] 
        &= \Pr[\mathsf{win} \gets \Hyb_1]
        \\
        &= \Pr[\mathsf{win} \gets \Hyb_2]
        \\
        &= \Pr[\mathsf{win} \gets \Hyb_3]
    \end{align*}
    because measuring the $\mathsf{prep}$ register to get $z$ produces $\ket{S_{x,z}}$ from the correct distribution in $\Hyb_1$, operations on disjoint registers commute in $\Hyb_2$, and a combination of $UU^\dagger = I$ and operations on disjoint registers commuting in $\Hyb_3$.

    Finally, we show that 
    \[
        \Pr[\mathsf{win} \gets \Hyb_3] \leq \Pr[\mathsf{win}\gets \Hyb_4] + \negl
    \]
    The only case in which the adversary can win in $\Hyb_3$ but not in $\Hyb_4$ is when $v_1 \in (S + x) \backslash (R + x_R)$:
    \[
        \Pr[\mathsf{win} \gets \Hyb_3] \leq \Pr[\mathsf{win}\gets \Hyb_4] + \Pr_{\Hyb_3}[v_1 \in (S + x) \backslash (R + x_R)]
    \]
    We show that $v_1 \in (S + x) \backslash (R + x_R)$ in $\Hyb_3$ with negligible probability. Consider the following sub-hybrids.
    \begin{itemize}
        \item $\Hyb_{3,1}$ moves the measurement of $x_R$ to \emph{before} initializing the adversary. This collapses the state in register $\calS$ to $\ket{R_{x_R, z_{R^\perp}}}$.
        \item $\Hyb_{3,2}$ changes $\adv_1$'s input. It instead receives $\Sim\big(\ket{R_{x_R, z_{R^\perp}}}, T + x_T, R^\perp + z_{R^\perp}\big)$ for uniformly random $T$, $R$, $x_T \in \co(T)$, $z_{R^\perp} \in \co(R^\perp)$, and $x_R \in \co(R)$ subject to $R\subset T$.\footnote{This step is the primary change from \cite{EC:BGKMRR24}'s proof.}
    \end{itemize}
    Since operations on disjoint registers commute,
    \[
        \Pr_{\Hyb_3}[v_1 \in (S + x) \backslash (R + x_R)] = \Pr_{\Hyb_{3,1}}[v_1 \in (S + x) \backslash (R + x_R)]
    \]
    Because $v_1 \in (S + x) \backslash (R + x_R)$ is efficiently checkable, the computational security of $(D_1,D_2)$ implies that
    \[
        \left|\Pr_{\Hyb_{3,1}}[v_1 \in (S + x) \backslash (R + x_R)] - \Pr_{\Hyb_{3,2}}[v_1 \in (S + x) \backslash (R + x_R)] \right| = \negl
    \]
    Finally, because the adversary's view in $\Hyb_{3,2}$ depends only on uniformly random $R_{x_R, z_{R^\perp}}$, $T + x_T$, and $R^\perp + z_{R^\perp}$,
    \begin{align*}
        \Pr_{\Hyb_{3,2}}[v_1 \in (S + x) \backslash (R + x_R)] 
        &\leq \max_{v_1 \in (T + x_T) \backslash (R + x_R)} \Pr_{S, x}[v_1 \in S + x] 
        \\
        &= \frac{2^{d_S} - 2^{2_R}}{2^{d_T} - 2^{d_R}} 
        \\
        &\leq 2^{-(d_T - d_S)}
        \\
        &= \negl
    \end{align*}
    
    Therefore 
    \[
        \Pr[\mathsf{win} \gets \Hyb_0] = \Pr[\mathsf{win} \gets \Hyb_3] \leq \negl + \negl = \negl
    \]
\end{proof}

\end{document}